\documentclass[11pt]{article}
\usepackage[margin=1in,a4paper]{geometry}
\usepackage{authblk}
\usepackage{amssymb}

\usepackage[utf8]{inputenc}
\usepackage{hyperref}

\usepackage{graphicx}
\usepackage{amsmath}
\usepackage{amsthm}
\usepackage{amsfonts}
\usepackage{algorithm}
\usepackage{algorithmic}
\usepackage{paralist} 
\usepackage{color}
\usepackage{mathtools}
\usepackage{xspace}
\usepackage{hyperref}
\usepackage{subcaption} 

\usepackage{thmtools}
\usepackage{thm-restate} 

\usepackage{xcolor}

\makeatletter
\newcommand*{\bx}{bx}
\newcommand*{\IfBold}{
  \ifx\f@series\bx
    \expandafter\@firstoftwo
  \else
    \expandafter\@secondoftwo
  \fi
}
\makeatother

\newcommand{\alg}[2][]{{\IfBold{\MakeUppercase{#2}}{\textup{\textsc{#2}}}}{#1}\xspace}
\newcommand{\trycolor}[1][]{\alg[#1]{TryColor}}
\newcommand{\tryrandomcolor}{\alg{TryRandomColor}}
\newcommand{\computeacd}{\alg{ComputeACD}}
\newcommand{\slackgeneration}[1][]{\alg[#1]{GenerateSlack}}
\newcommand{\putaside}[1][]{\alg[#1]{PutAside}}
\newcommand{\synchronizedcolortrial}{\alg{SynchColorTrial}}
\newcommand{\multitrial}[1][]{\alg[#1]{MultiTrial}}
\newcommand{\slackcolor}[1][]{\alg[#1]{SlackColor}}
\newcommand{\buddy}[1][]{\alg[#1]{Buddy}}
\newcommand{\friend}[1][]{\alg[#1]{Friend}}
\newcommand{\estimatesimilarity}[1][]{\alg[#1]{EstimateSimilarity}}
\newcommand{\jointsample}[1][]{\alg[#1]{JointSample}}
\newcommand{\estimatesparsity}[1][]{\alg[#1]{EstimateSparsity}}

\newcommand{\model}[1]{{\IfBold{\MakeUppercase{#1}}{\textup{\textsc{#1}}}}\xspace}
\newcommand{\LOCAL}{\model{Local}}
\newcommand{\CONGEST}{\model{Congest}}

\newcommand{\Ppolyclass}{\ensuremath{\mathbf{P}{/}\mathsf{poly}}\xspace}
\newcommand{\Pclass}{\ensuremath{\mathbf{P}}\xspace}
\newcommand{\eps}{\varepsilon}
\renewcommand{\epsilon}{\eps}

\DeclareMathOperator{\poly}{poly}

\definecolor{darkgreen}{rgb}{0,0.5,0}
\definecolor{darkblue}{rgb}{0,0,0.6}
\usepackage{hyperref}
\hypersetup{
    unicode=false,          
    colorlinks=true,        
    linkcolor=darkblue,          
    citecolor=darkgreen,        
    filecolor=magenta,      
    urlcolor=blue           
}
\usepackage[capitalize, nameinlink]{cleveref}
\Crefname{remark}{Remark}{Remarks}
\Crefname{observation}{Observation}{Observations}
\crefname{section}{Sec.}{Sec.}
\crefname{algorithm}{Alg.}{Alg.}

\newtheorem{theorem}{Theorem}
\newtheorem{lemma}{Lemma}
\newtheorem{claim}{Claim}

\newtheorem{definition}{Definition}

\newtheorem{proposition}{Proposition}

\DeclareMathOperator*{\Exp}{\mathbb{E}}

\DeclarePairedDelimiter{\abs}{\lvert}{\rvert}
\DeclarePairedDelimiter{\card}{\lvert}{\rvert}
\DeclarePairedDelimiter{\set}{\lbrace}{\rbrace}
\DeclarePairedDelimiter{\event}{\lbrack}{\rbrack}
\DeclarePairedDelimiter{\range}{\lbrack}{\rbrack}
\DeclarePairedDelimiter{\parens}{\lparen}{\rparen}
\DeclarePairedDelimiter{\ceil}{\lceil}{\rceil}

\newcommand{\knuthupuparrow}{\mathbin{\uparrow\uparrow}}

\newcommand{\bbN}{\mathbb{N}}
\newcommand{\bbZ}{\mathbb{Z}}

\newcommand{\BAD}{\ensuremath{\mathtt{BAD}}} 
\newcommand{\YES}{\ensuremath{\mathtt{true}}}
\newcommand{\NO}{\ensuremath{\mathtt{false}}}

\newcommand{\colSpace}{\mathcal{C}} 
\newcommand{\colspace}{\colSpace} 
\newcommand{\pal}{\Psi} 
\newcommand{\col}{\psi} 

\newcommand{\acset}{\mathcal{S}_{\mathrm{ac}}} 

\newcommand{\core}{I} 

\newcommand{\disc}{\bar{\eta}} 
\newcommand{\unev}{\eta} 
\newcommand{\barsigma}{\bar{\sigma}}
\newcommand{\spar}{\zeta} 
\newcommand{\lspar}{\rmsup{\spar}{\scriptscriptstyle[d]}} 
\newcommand{\gspar}{\rmsup{\spar}{\scriptscriptstyle[\Delta]}} 

\newcommand{\rmsup}[2]{#1^{\mathrm{#2}}}

\newcommand{\discC}{\rmsup{\disc}{C}} 

\newcommand{\pgen}{p_{\mathrm{g}}} 
\newcommand{\pdisj}{p_{\mathrm{s}}} 

\newcommand{\cext}{c_{\mathrm{e}}} 
\newcommand{\cant}{c_{\mathrm{a}}} 

\newcommand{\rmsub}[2]{#1_{\mathrm{#2}}}
\newcommand{\Vrm}[1]{\rmsup{V}{#1}}

\newcommand{\Vsparse}{\Vrm{sparse}}  
\newcommand{\Vsp}{\Vsparse}          
\newcommand{\Vdense}{\Vrm{dense}}    
\newcommand{\Vde}{\Vdense}           
\newcommand{\Vun}{\Vrm{uneven}}      
\newcommand{\Vst}{\Vrm{start}}       

\newcommand{\epsrm}[1]{\rmsub{\eps}{#1}}
\newcommand{\eacd}{\epsrm{ac}}  
\newcommand{\espa}{\epsrm{sp}}  
\newcommand{\ehat}{\hat{\eps}}  

\newcommand{\DeloC}{\ensuremath{\mathrm{\Delta1C}}\xspace}
\newcommand{\degoLC}{\ensuremath{\mathrm{D1LC}}\xspace}
\newcommand{\degoC}{\ensuremath{\mathrm{D1C}}\xspace}

\newcommand{\HFset}{\mathcal{H}} 
\newcommand{\HFpwi}{\HFset_{\mathrm{pwi}}} 
\newcommand{\inj}[3][]{#2|^{#1}_{#3}} 
\newcommand{\collide}[4][]{#2\wedge^{#1}_{#3}#4} 
\newcommand{\hitsymbol}{\mathbin{\neg}}
\newcommand{\hit}[4][]{#2\hitsymbol^{#1}_{#3}#4} 
\newcommand{\samp}{\sigma} 
\newcommand{\out}{\lambda} 
\newcommand{\rhfSpace}{\mathcal{U}} 
\newcommand{\rhfspace}{\rhfSpace} 

\newcommand{\famsize}{F} 

\newcommand{\smin}{s_{\min}}
\newcommand{\sminpow}{\rho}

\newcommand{\bw}{b} 
\newcommand{\Sampler}{\textsc{Samp}} 
\newcommand{\enc}{\mathtt{enc}} 

\graphicspath{{./figures/}{./figs/}{./}}

\begin{document}

\title{Overcoming Congestion in Distributed Coloring\thanks{This paper incorporates results from the technical report \cite{HNT21} on adapting \LOCAL algorithms to \CONGEST. This excludes the other results in \cite{HNT21}, which were refactored in \cite{HKNT21}.}}

\author[1]{Magn\'us M. Halld\'orsson}
\author[1]{Alexandre Nolin}
\author[2]{Tigran Tonoyan}
\affil[1]{ICE-TCS, Department of Computer Science, Reykjavik University, Iceland} 
\affil[2]{Krisp Technologies Inc., Armenia} 

\affil[ ]{ \href{mailto:mmh@ru.is}{mmh@ru.is}; \href{mailto:alexandren@ru.is}{alexandren@ru.is}; \href{mailto:ttonoyan@gmail.com}{ttonoyan@gmail.com}}

\date{}

\maketitle

\begin{abstract}
    We present a new technique to efficiently sample and communicate a large number of elements from a distributed sampling space.
    When used in the context of a recent \LOCAL algorithm for $(\operatorname{degree}+1)$-list-coloring ($\degoLC$), this allows us to solve $\degoLC$ in $O(\log^5 \log n)$ \CONGEST rounds, and in only $O(\log^* n)$ rounds when the graph has minimum degree $\Omega(\log^7 n)$, w.h.p. 
    
    The technique also has immediate applications in testing some graph properties locally, and for estimating the sparsity/density of local subgraphs in $O(1)$ \CONGEST rounds, w.h.p.
\end{abstract}

\section{Introduction and Related Work}
\label{sec:intro}


We explore ways to reduce bandwidth, particularly for the fundamental vertex coloring problem.
Bandwidth is the key difference between the otherwise similar models of locality that are \LOCAL and \CONGEST: while nodes can send messages of arbitrary size in \LOCAL (and thus, the round complexity of a problem only depends on how far in the graph its nodes need to see), in \CONGEST the nodes are restricted to messages of size $O(\log n)$, where $n$ is the number of nodes in the graph. While some classical algorithms designed without bandwidth in mind nonetheless immediately work in both models, this is not true of many recent algorithms for distributed coloring.

In this paper, we introduce a technique to implement some general sampling and estimation tasks in \CONGEST. At its heart are families of hash functions we call \emph{representative} for having certain statistical properties. The technique improves on the bandwidth cost of a na\"ive approach by optimizing the use of randomness, taking ideas from the pseudorandomness literature. In the context of distributed coloring, the technique allows us to adapt crucial parts of recent randomized \LOCAL algorithms to \CONGEST. Through this and some additional ideas, we construct randomized \CONGEST algorithms for some of the most common distributed coloring problems that almost match the complexity of the current best \LOCAL algorithms for the same problems. Aside from our results for vertex coloring, our technique has direct applications of independent interest related to testing for the presence of some graph structure, and constructing structural decompositions known as \emph{almost-clique decompositions} in \CONGEST.

The key idea behind our technique is a combination of using hashing and optimizing the amount of randomness we use when selecting a hash function. In bandwidth-constrained contexts, hashing is a natural hammer for a number of nails. It allows to exchange information about sparse data living in a very large space by exchanging much smaller images. Random hash functions are also useful tools to sample elements, by taking as sample the elements that hash to certain values. However, the full description of an arbitrary function requires a lot of bits, necessarily more than the data we want to hash. Only by making the hash functions we use less random can they become useful tools in \CONGEST. Overall, using hashing in \CONGEST is a balancing act between using enough randomness for the random hash functions to have the statistical properties we need, and using randomness efficiently and sparingly so as to make the hash functions communicable within our bandwidth constraints.

When analyzing recent \LOCAL algorithms for vertex coloring focusing on the bandwidth use of each of their steps, two steps in particular stand out. One is related to the computation of so-called almost-clique decompositions. In such decompositions, nodes decide to join their neighbors in so-called \emph{almost-cliques} depending on how similar their neighborhood is to that of their neighbors. Since in a graph of maximum degree $\Delta$, the description of a node's neighborhood requires up to $\Delta \log n$ bits, a na\"ive approach would require $\Delta$ \CONGEST rounds. 
The other seemingly high-bandwidth step has to do with a procedure named \multitrial, which has nodes send multiple colors (up to $\Theta(\log n)$) to their neighbors. With colors living in a set $\colspace$, describing an arbitrary set of this many colors takes $\log |\colspace| \log n$ bits, i.e., $\log |\colspace|$ rounds.

Both steps have the shared quality of essentially reducing to some sampling task. In the case of the almost-clique decomposition, two nodes can get a sense of how similar their neighborhoods are by sending each other random elements from their respective neighborhoods.
In the case of \multitrial, the goal of each node is to sample random colors jointly with its neighbors such that each color has a good chance of being both a valid color for it and not simultaneously sampled by one of its neighbors.
This and the fact that some sort of random sampling also appears in a variety of other randomized algorithms motivates looking for an efficient implementation in \CONGEST.

The cost of communicating a random sample is intrinsically linked to how random it is: for example, it requires $10$ times less communication to communicate $10$ random elements if they are guaranteed to be all equal instead of being all independent. Thus, it might be tempting to use less randomness to save on communication in a distributed setting. But modifying a working algorithm by making it use less randomness runs the risk of skewing the probabilities to a point where the algorithm no longer works. The field of pseudorandomness has come up with techniques to save on randomness while keeping an algorithm functional, for example in the context of repeating a randomized algorithm to boost the probability of computing the correct answer. Ideas related to pseudorandomness have also previously made their way in fields focusing on the communication cost of algorithms, notably in the form of a seminal result in 2-party communication complexity known as Newman's theorem~\cite{Newman91}. In our setting where we use hash functions as tools for sampling elements, techniques from pseudorandomness allow us to construct families of hash functions that strike a balance between being random enough to produce useful samples, but small enough for our communication constraints.

One of the most general versions of the vertex coloring problem is the \emph{(degree+1)-list-coloring} problem ($\degoLC$). In this version, each vertex $v$ is given a list of $d_v + 1$ colors at the beginning of the algorithm, where $d_v$ is the degree of $v$. Each node must then color itself with a color from its list that is distinct from the colors adopted by its neighbors. Compared to versions of the problem where nodes all receive a list of $\Delta+1$ colors, where $\Delta$ is the maximum degree in the graph, $\degoLC$ requires to find procedures that can be parameterized to work with nodes that differ greatly in the number of colors they can choose from. Our techniques using pseudorandom hash functions allows such parameterization, by using different sets of hash functions depending on the size of the space we are trying to sample from.
Once all steps reducible to a sampling task are implemented using our pseudorandom family of hash functions, few steps of recent randomized \LOCAL algorithms remain to be adapted to work in \CONGEST. We give a complete \CONGEST adaptation of a recent \LOCAL algorithm for $\degoLC$. On graphs of large minimum degree, the resulting algorithm matches the ultrafast complexity of the \LOCAL algorithm it draws from. On graphs containing nodes of lower degree, our algorithm has a higher complexity but remains of order polynomial in $\log \log n$.

As our main tool for tackling $\degoLC$ in \CONGEST is a general technique for different types of sampling tasks, and sampling is extensively used in randomized algorithms, it might find uses in a variety of other problems.
In fact, we give two simple direct applications in the context of subgraph detection.
Our $O(1)$ \CONGEST algorithm for computing an almost-clique decomposition might also prove useful in problems other than vertex coloring.

As our main tool is only shown to exist through an existential proof, our algorithms are not uniform in their default form, in the sense that they require that the nodes either perform massive computation exploring the set of all hash functions, or are given some common advice bits that only depends on the size of the input. We provide explicit, uniform implementations of the subroutines crucial to our vertex coloring results. These implementations are ad hoc, but can be taken as indications that explicit constructions of our general tool of representative hash functions might be possible.

\subsection{Related Work}

The literature of distributed graph coloring is vast and we only mention those with direct implications for our work.
In the paper introducing the \LOCAL model \cite{linial92}, Linial showed that $(\Delta+1)$-coloring constant-degree graphs requires $\Omega(\log^* n)$ rounds, and gave a matching deterministic algorithm. This remains the only lower bound known when this many colors are available. The best upper bounds known on deterministic algorithms in general graphs are $O(\log^2 \Delta \log n)$ \cite{GK21} and $\tilde{O}(\sqrt{\Delta}) + O(\log^* n)$ \cite{Barenboim16}.

The \textsf{MultiTrial} technique was introduced in \cite{SW10} and further developed in \cite{EPS15} and \cite{CLP20}. Slack generation via sparsity was introduced in \cite{EPS15} (though it traces back to \cite{Reed98} in graph theory), where it was used to give ultrafast algorithms for edge coloring graphs of high degree. The shattering framework for distributed algorithms was proposed in \cite{BEPSv3}. 
The almost-clique decomposition was introduced to distributed computing in \cite{HSS18} (while a similar one was known in graph theory \cite{MolloyR98}), leading to an $O(\sqrt{\log n})$-round algorithm for $(\Delta+1)$-coloring. An implementation of ACD with random sampling was proposed in \cite{ACK19}. The best randomized $(\Delta+1)$-coloring algorithm known, given in \cite{CLP20}, has complexity $O(\log^3 \log n)$. 


For the $(\deg+1)$-list coloring (\degoLC) problem, the best bound until recently was $O(\log n)$ of the early algorithms of 
\cite{johansson99,alon86,luby86} and the more refined bound of $O(\log \Delta + \poly(\log\log n))$ \cite{BEPSv3}.
This year, the bound was improved to $O(\log^3 \log n)$, which drops to $O(\log^* n)$ when all nodes have degree $\Omega(\log^{7} n)$ \cite{HKNT21}.

All the works above (except the recent \cite{GK21}) were stated for the \LOCAL model, while the $O(\log n)$-round algorithms also run in \CONGEST. The complexity of $(\Delta+1)$-list-coloring in \CONGEST was recently improved to $O(\log^5 \log n)$ rounds in \cite{HKMT21}.

For the Congested Clique, $O(1)$-round algorithms are known for $(\Delta+1)$-coloring, both randomized \cite{CFGUZ19} and deterministic \cite{CDP20}. An earlier $O(1)$-round algorithm was given for semi-linear MPC \cite{ACK19}.

$(1+\eps)\Delta$ coloring, $(2\Delta-1)$-edge coloring, and $(1+\eps)\Delta^2$ distance-2 coloring were recently shown to respectively admit an $O(\log^3 \log n)$, $O(\log^4 \log n)$, and $O(\log^4 \log n)$ rounds algorithm in \CONGEST \cite{HN21}. When $\Delta \geq \log^{1+1/\log^*n} n$, the complexities drop to $O(\log^* n)$ and the edge-coloring algorithm can restrict itself to $(1+\eps)\Delta$ colors. The efficiency of these algorithms at high degrees comes from a pseudorandom construct called representative set families. Such families are built so that a random member is likely to intersect any large fraction of the space. This allows nodes to efficiently sample up to $\Theta(\log n)$ colors in $O(1)$ rounds when given a constant fraction of their degree as extra colors, making $O(\log^* n)$ algorithms possible. The technique has an important limit, however: it does not work when the colors live in a color space much larger than the nodes' degrees. This prevents the technique from being useful in list-coloring settings, and when nodes have an amount of slack that is comparable to their degree late in the algorithm, but much smaller than their original degree. We build on the technique in this paper, and tackle scenarios in which representative sets could not be applied. Obtaining this more general result requires a significant leap, as exchanging information about and sampling from elements living in a very large universe efficiently requires a high level of succinctness 
in our communication.


Distributed property testing was introduced in \cite{BP11} and formalized in \cite{CFSV19}. The usual model is to distinguish graphs that satisfy a given property (e.g., triangle-freeness), from those for which an $\epsilon$-fraction of the edges must be removed for the property to hold. Distributed testing algorithms were given for detecting triangles in \cite{CFSV19} and $C_4$ in \cite{FRST16}, with the best known round complexity of $O(1/\eps)$ obtained for both problems (and other cycles) in \cite{FO19}.


\section{Results}

\paragraph{Efficient sampling and estimation.}
We give a communication-efficient procedure for two parties each possessing a set to estimate how similar their sets are, and sample an element in their intersection or their difference (\cref{sec:estimate-and-sample}). The technique also works with more parties, allowing, e.g., a party to sample an element in the difference between her set and the union of all her neighbors' sets. The technique is quite general, and might find applications in problems other than those studied in this paper.

\paragraph{Coloring.}
We bring to \CONGEST recent \LOCAL $\poly(\log \log n)$-round randomized algorithms for coloring, at the cost of a moderate increase in complexity -- our algorithm uses $O(\log^5 \log n)$ \CONGEST rounds in general, up from $O(\log^3 \log n)$ \LOCAL rounds. Our algorithm is an adaptation of a recent \LOCAL algorithm for $\degoLC$ of~\cite{HKNT21}.

\begin{restatable}{theorem}{congestDegolcTheorem}\label{thm:congestDegolc}
The \degoLC problem can be solved w.h.p.\ in $O(\log^5\log n)$ rounds in the \CONGEST model with $\log n$ bandwidth. When all nodes have degree at least $\log^7 n$, the algorithm uses only $O(\log^* n)$ rounds.
\end{restatable}

This coloring algorithm uses only polynomial local computation.

The round complexity reduces to $O(\log^3 \log n)$ when the size of the color space $\colspace$ (and therefore, the degrees) is of order $\poly(\log n)$.
This, in combination with the 
$O(\log^* n)$ complexity on large degree nodes, immediately yields an $O(\log^3 \log n)$ algorithm for $(\operatorname{degree}+1)$-coloring (\degoC).
Note that this also improves on the state of the art for the $(\Delta+1)$-coloring (\DeloC) problem in \CONGEST.

\begin{restatable}{corollary}{betterloglogdelta}\label{cor:congestDeloC}
\degoC can be solved w.h.p.\ in $O(\log^3 \log n)$ rounds in the \CONGEST model.
\end{restatable}

\paragraph{Uniform implementations.}
The implementation that follows from our main technique is non-uniform in the sense that it requires that the nodes perform very large computations locally, or that they have access to some advice only dependent on the size of the input (similar to how the complexity class \Ppolyclass is enhanced compared to \Pclass). To reduce the total computational demand to polynomial, we provide alternative uniform implementations of our main procedures in \cref{sec:uniform}.

\paragraph{Other results.}
On our way to proving our results for vertex coloring, we give an algorithm for computing a $(\operatorname{degree}+1)$ almost-clique decomposition (\cref{sec:acd}). Our general technique for sampling and estimation also has some immediate applications in testing for the presence certain graph structures, e.g., triangle-rich neighborhoods (\cref{sec:sparsity,sec:triangles,sec:cycles}).

\section{Congestion-Reducing Techniques}
\label{sec:congest}

\subsection{Representative Hash Functions}

The crux of our results is a procedure for communicating parties to estimate the intersection or difference of sets they keep and/or sample elements in that intersection or difference. We do so through hashing, using a family of hash functions we call \emph{representative} due to their statistical properties. For a given parameter $b$, the family is engineered to distort the probabilities of some events by at most $\exp(-\Omega(b))$ compared to fully random hash functions, while being of small enough size $\exp(O(b))$. With $b$ chosen within a constant multiplicative factor of the available bandwidth, i.e., $b \in \Theta(\log n)$, this enables sending the index of a function in the family in a constant number of \CONGEST rounds, while only introducing a manageable distortion compared to a fully random hash function.

\paragraph{Intuition.}
Suppose two parties have access to a shared source of randomness to pick a fully random hash function $h$ from a universe $\rhfspace$ to $[\out]$, without communicating. Having access to such a hash function offers several possibilities. In particular, the parties may now communicate about an element $x \in \rhfspace$ through its image $h(x)$. Suppose the communicating parties each possess a subset of $\rhfspace$, respectively $X$ and $Y$. To pick a random element in $X$, the node possessing $X$ can rely on the randomness of the hash function for the selection process: set a threshold $\samp$, consider all the elements of $X$ that hash to a value $\leq \samp$, and pick one of those low-hashing elements at random. When sampling elements jointly, the parties can ensure that they choose two distinct elements $(x,y) \in X \times Y, x \neq y$ by ensuring that $h(x) \neq h(y)$. To sample an element in $X \setminus Y$, it suffices to pick an element in $X \setminus h^{-1}(h(Y))$, i.e., an element of $X$ that hashes to a value that no element of $Y$ hashes to. To pick an element in the intersection $X \cap Y$, the parties may look at the intersection $h(X) \cap h(Y)$ and take the preimages of a hash in the intersection. When bandwidth is limited, the parties can adapt to this constraint by adjusting the threshold $\samp$: $\samp$ bits suffice for each party to encode, for each value $\leq \samp$, whether it has an element hashing to it.

The hash function introduces errors: the set $h(X)\setminus h(Y)$ can be empty though $X\setminus Y$ is not, due to collisions; the set $h(X) \cap h(Y)$ might be non-empty even when $X \cap Y$ is; and our elements of interest might hash to values $>\samp$, resulting in the parties missing them. But with the right ratios between the sizes of the sets, the size of the output space of the hash function ($\out$), and the size of the observation window ($\samp$), it can be argued that only some amount of errors occurs with the needed probability. 

To apply these ideas in \CONGEST, it only remains for the parties to be able to sample and communicate a random hash function, which is achieved by building a small set of hash functions with nearly the same statistical properties as the set of all hash functions from $\rhfspace$ to~$[\out]$.

\paragraph*{Notations.}
For a set $\rhfspace$ and a number $\out\in \bbN$, let $[\out]^{\rhfspace}$ denote the set of all functions from $\rhfspace$ to $[\out]=\{1,\dots,\out\}$. For a function $h$, sets $A,B \subseteq \rhfspace$, and number $\samp$, we define:
\begin{itemize}
    \item $\inj[\leq\samp]{A}{h} = h^{-1}([\samp])\cap A$,
    \item $\collide[\leq \samp]{A}{h}{B} = \set*{\col \in A : h(\col)\in [\samp]\cap h(B\setminus \set{\col})}$,
    \item $\hit[\leq \samp]{A}{h}{B} = (A|_h^{\leq\samp}) \setminus (\collide[\leq \samp]{A}{h}{B})$.
\end{itemize}

Intuitively, $\inj[\leq\samp]{A}{h}$ are the elements of $A$ that hash to a value at most $\samp$ through $h$, $\collide[\leq \samp]{A}{h}{B}$ is the subset of $\inj[\leq\samp]{A}{h}$ that is in collision with some element of $B$, i.e., the elements $x\in \inj[\leq\samp]{A}{h}$ s.t.\ there exists a $x'\in B\setminus \set{x}, h(x') = h(x)$.
$\hit[\leq \samp]{A}{h}{B}$ are the elements of $A$ that hash to a value $\leq \samp$ through $h$ that no distinct element in $B$ hashes to (note that $B$ may contain $A$ or a subset of $A$). The definitions of the sets are most clear when $B=A$: $\collide[\leq \samp]{A}{h}{A}$ are the elements of $A$ that hash to at most $\samp$ through $h$ and collide with another element of $A$; $\hit[\leq \samp]{A}{h}{A}$ are the elements of $A$ that hash to at most $\samp$ through $h$ and do not collide with another element of $A$. Note that $\inj[\leq\samp]{A}{h}$, $\collide[\leq\samp]{A}{h}{B}$, and $\hit[\leq\samp]{A}{h}{B}$ are subsets of the domain of $h$, not its codomain.

If $h$ were fully random, we would expect the size of $\inj[\leq\samp]{A}{h}$ to be within a constant factor of $\samp\card{A}/\out$ w.p.\ $1-\exp(-\Omega(\samp\card{A}/\out))$. Also, with a fully random $h$ and $\out$ sufficiently large w.r.t.\ $A$ and $B$, we would expect at most $\card{A}\card{B}/\out$ elements of $A$ to be in collision with an element of $B$. Our goal is to maintain a relaxed version of these probabilistic guarantees while restricting the space of random hash function we select from, so that communicating the index of a selected function is feasible in $O(1)$ messages.

\begin{figure}[h]
    \centering
    \includegraphics[page=1,width=0.7\linewidth]{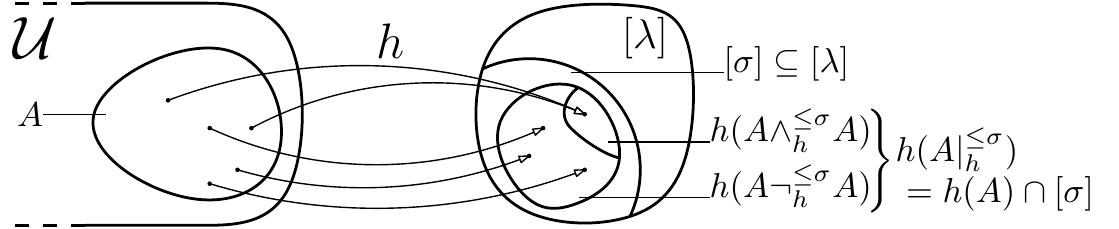}
    \caption{Example of our notation on a symmetric example ($B=A$). $\inj[\leq\samp]{A}{h}$ is the part of $A$ hashing to $\samp$ or less, $\collide[\leq \samp]{A}{h}{A}$ are the elements of $\inj[\leq\samp]{A}{h}$ that collide through $h$, and $\hit[\leq \samp]{A}{h}{A}$ is the rest of $\inj[\leq\samp]{A}{h}$.}
    \label{fig:rhf_blobs_sym}
\end{figure}

\begin{figure}[b]
    \centering
    \includegraphics[page=2,width=0.7\linewidth]{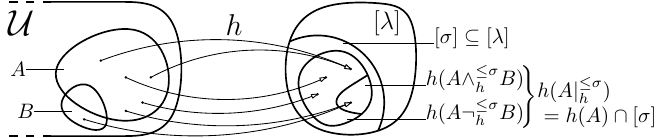}
    \caption{Example of our notation on an asymmetric example ($B\neq A$). $\collide[\leq \samp]{A}{h}{B}$ and $\hit[\leq \samp]{A}{h}{B}$ are the natural non-symmetric generalizations of $\collide[\leq \samp]{A}{h}{A}$ and $\hit[\leq \samp]{A}{h}{A}$ where we focus on collisions between elements of $A$ and $B$ instead of within $A$. $A$ and $B$ can overlap arbitrarily.}
    \label{fig:rhf_blobs_asym}
\end{figure}

For intuition and easier proofs later, a few elementary properties of set operators $\inj[\leq \samp]{}{h}$, $\collide[\leq \samp]{}{h}{}$, and $\hit[\leq \samp]{}{h}{}$ are given in \cref{prop:rhf-elementary}. When $\samp$ is clear from the context, we omit the superscript $\scriptstyle \leq \samp$ and simply write $\inj{A}{h}, \collide{A}{h}{B}$, and $\hit{A}{h}{B}$.

\begin{proposition}
\label{prop:rhf-elementary}
\begin{align}
     \forall A &: & \card*{h(\collide[\leq \samp]{A}{h}{A})} &\leq \card*{\collide[\leq \samp]{A}{h}{A}}/2\ ,
     \label{eq:rhf-collide}\\
     \forall A,B \;&\textrm{ s.t. } A\subseteq B:& \card*{h(\hit[\leq \samp]{A}{h}{B})} &= \card*{\hit[\leq \samp]{A}{h}{B}}\ ,
     \label{eq:rhf-hit}\\
    \forall A,B,C \;& \textrm{ s.t. } B \subseteq C: & (\collide[\leq \samp]{A}{h}{B}) &\subseteq (\collide[\leq \samp]{A}{h}{C})  \;\textrm{ and thus }\; (\hit[\leq \samp]{A}{h}{C}) \subseteq (\hit[\leq \samp]{A}{h}{B})\ .
    \label{eq:rhf-inclusion}
\end{align}
\end{proposition}
\begin{proof}
\Cref{eq:rhf-collide} follows from the fact that for each element $x \in \collide[\leq \samp]{A}{h}{A}$, there exists an element $x' \in A\setminus \set{x}, h(x') = h(x)$, which by definition of $\collide[\leq \samp]{A}{h}{A}$ implies that $x' \in \collide[\leq \samp]{A}{h}{A}$. So every hash $y\in h(\collide[\leq \samp]{A}{h}{A})$ has at least two preimages, giving the result.

\Cref{eq:rhf-hit} follows from every element of $\hit[\leq \samp]{A}{h}{B}$ having a unique hash value among elements of $B$, of which $\hit[\leq \samp]{A}{h}{B}$ is a subset.

\Cref{eq:rhf-inclusion} describes the simple fact that the set of elements in collision with elements from a set $B$ is smaller than (and included in) the set of elements in collision with elements from the set $C$, where $C$ is a superset of $B$. Conversely, any element not in collision with all the elements of $C$ is necessarily not colliding with any element of $B$. 
\end{proof}

We now show the existence of a small (compared to $[\out]^\rhfspace$) family of hash function such that a random element from this family behaves similar to a fully random hash function w.r.t\ to the sets $\hit[\leq\samp]{A}{h}{B}$, $\collide[\leq\samp]{A}{h}{B}$, and $\hit[\leq\samp]{A}{h}{B}$, for all $A$ and $B$ in a given size range. The proof is of the probabilistic method type, and bears resemblance to other arguments in the field of pseudorandomness or more generally aimed at reducing the amount of random bits used in a task, such as Newman's theorem~\cite{Newman91} and recent efforts to bring ultrafast distributed coloring algorithms in \CONGEST \cite{HN21}. This construction is the basis of our algorithms for efficiently estimating and sampling in \CONGEST, \estimatesimilarity (\cref{alg:similarity}) and \multitrial (\cref{alg:multitrial}).
Note that using simpler objects like $k$-wise independent hash functions would not suffice here. Indeed, such hash functions do not ensure the required statistical properties with sufficient probability when $k$ is low, and increasing $k$ to the number of elements we generally need to hash ($\max(\Delta,\log n)$) to get the needed probability prohibitively increases the cost of describing a hash function in \CONGEST.

\begin{lemma}[Representative hash functions]
\label{lem:representative_hash_functions}

Let $\alpha,\beta,\nu\in (0,1)$ and $\out\in \bbN$ be s.t.\ $\alpha \leq \beta$ 
and $\out \geq \max(45\alpha^{-1},3\alpha^{-1}\beta^{-2})\ln(12/\nu)$, 
and let $\rhfspace$ be a finite set. There exists a family of $\famsize=\Theta\left(\beta \out\nu^{-1} \log\card{\rhfspace}\right)$ hash functions $\{h_i\}_{i\in[\famsize]}\subseteq [\out]^{\rhfspace}$  and $\samp\leq \out$, $\samp \in \Theta\left(\beta^{-2}\alpha^{-1}\log(1/\nu)\right)$, such that for every $A,B \subseteq \rhfspace$ with $\card{A}, \card{B} \in [0,\beta \out]$, at least $(1-\nu)\famsize$ of the hash functions $h$ satisfy
\begin{align*}
\card*{\inj[\leq \samp]{A}{h}} &\in \frac {\samp \card{A}}{\out} \cdot \range*{1-\beta,1+\beta}
&\textrm{ and }\qquad & 
    \card*{\collide[\leq \samp]{A}{h}{B}} \leq \frac {2\samp \card{A}}{\out} \beta
& \textrm{when $\card{A} \geq \alpha \out$~,}
\\
    \card*{\inj[\leq \samp]{A}{h}} &\leq \samp \alpha\cdot (1+\beta)
&\textrm{ and }\qquad &
    \card*{\collide[\leq \samp]{A}{h}{B}} \leq 2\samp \alpha \beta
    & \textrm{when $\card{A} < \alpha \out$~.}
\end{align*}

\end{lemma}

To prove \cref{lem:representative_hash_functions}, we first prove the following claim. We only consider sets $A,B\subseteq \rhfspace$ satisfying $|A|,|B|\in [0,\beta\out]$. A hash function is \emph{$(A,B)$-good} if it satisfies the requirement of the lemma for a given pair $(A,B)$.
We bound the probability that a random function is $(A,B)$-good, for a fixed pair $(A,B)$. 
\begin{claim}
Let $h\in [\out]^{\rhfSpace}$ be chosen uniformly at random. Then $\Pr\event*{h \text{ is }(A,B)\text{-good}}\ge 1-\nu/2$.
\end{claim}
\begin{proof}
We prove the result when $\card{A} \geq \alpha\out$, the case when $\card{A} < \alpha\out$ being analogous.

For $x \in A$, let $X_x,Y_x$ be  indicator r.v.'s such that $X_x=1$  iff $h(x) \in [\samp]$, and $Y_x=1$ iff $h(x) \in [\samp]$ and $\exists x' \in B\setminus \set{x}$, $h(x)=h(x')$.
Let $X = \sum_{x \in A} X_x$ and $Y = \sum_{x \in A} Y_x$. Note that $X = \card*{\inj[\leq \samp]{A}{h}}$ and $Y=\card*{\collide[\leq \samp]{A}{h}{B}}$.
We have:

\[\Exp\event*{X_x} = \samp/\out \qquad\qquad\text{and}\qquad\qquad \Exp\event*{Y_x} = (\samp/\out)\cdot \parens*{1-\parens*{1-1/\out}^{\card{B}-\tau} } ~,\]
where $\tau=1$ if $x \in B$ and $\tau=0$ otherwise. Thus, letting $\mu=\Exp\event*{X}$, we have $\mu=\samp\card{A}/\out\ge \alpha \samp$.
Using the inequality $1-kx\le (1-x)^k$   
(for $n\in \bbZ_+$, $x \in [0,1]$),  
we have:
\[
 \parens*{1-1/\out}^{\card{B}-\tau}\ge  1-(\card{B}-\tau)/\out \ge 1-\beta\ ,
\]
which implies that $\Exp\event*{Y_x}\le \beta\Exp\event*{X_x}$, and hence $\Exp\event*{Y}\le \beta\mu$.

Note that $h$ being $(A,B)$-good corresponds to $\card*{X-\Exp\event*{X}}\le \beta \mu$ and $\card{Y}\le 2\beta \mu$ -- where $2\beta \mu \geq 2\Exp\event*{Y}$. We argue that both hold with sufficient probability through concentration inequalities.

For the first inequality, the independence of the $X_x$ r.v.'s implies that we can apply Chernoff (\cref{lem:basicchernoff}). Thus:
\[\Pr\event*{ \card*{X-\Exp\event*{X}}> \beta \mu} \leq 2 \exp(-\beta^2\mu/3)~.\]

It suffices that $\mu \geq 3\beta^{-2}\ln(8/\nu)$ for $\card*{X-\Exp\event*{X}} \geq \beta \mu$ to hold w.p.\ $\geq 1-\nu/4$. As we assumed $\card{A}\geq \alpha \out$, we have $\mu = \samp \card{A}/ \out \geq \alpha \samp$. Therefore, $\samp \geq 3\beta^{-2}\alpha^{-1}\ln(8/\nu)$, i.e., $\samp \in \Theta(\beta^{-2}\alpha^{-1}\log(1/\nu))$, suffices for the first inequality to hold w.p.\ $1-\nu/4$.

For the second inequality, as $\card{B}\leq\beta\out$, notice that $h(B)$ covers less than a $\beta$ fraction of the hash space below $\samp$ in expectation.

By the same analysis as the one we did just above with $A$ (applying \cref{lem:basicchernoff}) we obtain that at most $\sqrt[3]{2} \beta \samp$ elements of $B$ hash to a value $\leq \samp$, i.e., $\card*{\inj[\leq \samp]{B}{h}} \leq \sqrt[3]{2}\beta \samp$, w.p.\ $1-\nu/12$, when $\samp \geq 45 \beta^{-1}\ln(12/\nu)$ (we use $(\sqrt[3]{2}-1)^{-2}\leq 15$). We condition on this event, as well as on at most $\sqrt[3]{2} \card{A} \samp/\out$ elements of $A$ hashing to $\leq \samp$, which holds w.p.\ $1-\nu/12$, when $\samp \geq 45 \alpha^{-1}\ln(12/\nu)$. More precisely, we fix the subsets $B'= \inj[\leq \samp]{B}{h} \subseteq B$ and $A' = \inj[\leq \samp]{A}{h} \subseteq A$ of respective size at most $\sqrt[3]{2}\beta \samp$ and $\sqrt[3]{2}\card{A} \samp/\out$ containing the elements hashing to values less than $\samp$ (and excluding those hashing higher).

We now analyze $Y=\sum_{x \in A} Y_x$ under this conditioning. For each $x\in A\setminus A'$, we now have that $Y_x=0$ w.p.\ $1$, so $Y=\sum_{x \in A'} Y_x$. Under the conditioning, the value $h(x)$ of an element $x \in A'$ is now picked uniformly at random in $[\samp]$ independently of other elements' values. Hence, for an element $x\in A'$, $Y_x=1$ w.p.\ $\leq \sqrt[3]{2} \beta$ even conditioned on arbitrary random choices for the other $h(y), y\in (B'\cup A') \setminus \set{x}$. Therefore, $\Exp\range*{Y}\leq \sqrt[3]{2}\beta\card{A'}\leq \sqrt[3]{4}\beta\card{A}\samp/\out$, and by the martingale inequality (\cref{lem:chernoff}), \[\Pr\event*{Y>2\beta\card{A}\samp/\out} \leq \exp(-(\sqrt[3]{2}-1)^2\sqrt[3]{4}\beta\card{A}\samp/(3\out)) \leq 1-\nu/12~,\]
where the last step follows from assuming $\samp \geq 45 \alpha^{-1}\beta^{-1}\ln(12/\nu)$. Taking into account the previous conditioning of probability $1-\nu/6$ (setting $A'$ and $A'$), $\Pr\event*{Y>2\beta\card{A}\samp/\out}$ holds w.p.\ $1-\nu/4$.

Putting everything together, the two inequalities hold simultaneously, i.e., $h$ is $(A,B)$-good, w.p.\ $\geq 1-\nu/2$. Note that throughout the analysis, $\samp \leq \out$ was assumed, which constrains $\out$ in terms of $\alpha$, $\beta$, and $\nu$, as in the statement of the lemma.
\end{proof}

\begin{proof}[Proof of \cref{lem:representative_hash_functions}]
Let  $h_1,\ldots,h_\famsize\in [\out]^{\rhfspace}$ be $\famsize$ functions, chosen independently and uniformly at random. For fixed sets $A,B$, let $X_i=1$ if $h_i$ is not $(A,B)$-good,  otherwise $X_i=0$; by the claim above, $\Pr\event*{X_i}\le \nu/2$. By Chernoff (\cref{lem:basicchernoff}), the probability that more than $\nu \famsize$ of them fail to be $(A,B)$-good is
$ \Pr\event*{\sum_{i\in[\famsize]} X_i \geq \nu \famsize} \leq e^{-\nu \famsize/6 }$. There are at most $|\rhfspace|^{\beta \out+1}$ choices for each of the subsets $A$ and $B$, so at most $|\rhfspace|^{2\beta \out+2}$ choices for the pair $(A,B)$. By the union bound, the probability that there are $\nu \famsize$ functions that are not $(A,B)$-good for some pair $(A,B)$ is at most $|\rhfspace|^{4\beta \out}e^{-\nu \famsize/6}<1$, assuming $\famsize>(24\beta \out/\nu)\log |\rhfspace|$. Thus, there is a family of $\famsize$ hash functions such that for every pair $(A,B)$, at least $(1-\nu)\famsize$ hash functions from the family are $(A,B)$-good. 
\end{proof}

\subsection{Estimation and Sampling of Set Intersection, Union, Difference}
\label{sec:estimate-and-sample}

Representative hash function immediately give an efficient way for two nodes $u$ and $v$ possessing two sets $S_u$ and $S_v$ to estimate the size of the intersection $\card{S_u \cap S_v}$ with an accuracy $\eps\card{S_u \cap S_v}$, as long as $S_u \cap S_v$ is a large enough fraction of $S_u \cup S_v$. At a high level, the idea is quite natural: estimate the size of the intersection $\card{S_u \cap S_v}$ through the size of the intersection $\card{h(S_u) \cap h(S_v)}$, which itself is approximated by the subset $\card{h(S_u) \cap h(S_v) \cap [\samp]}$. This, of course, omits a few details, and we give the formal statement in \cref{lem:estimatesimilarity} and its proof. The same idea can be used to sample elements in the intersection rather than estimating its size, by having nodes pick as elements the preimages of a random element in $h(S_u) \cap h(S_v) \cap [\samp]$ (see \cref{lem:jointsample}).

\begin{algorithm}[H]\caption{{\estimatesimilarity}$(\eps)$, for edge $e=uv$, with sets $S_u,S_v$}
\label{alg:similarity}
  \begin{algorithmic}[1]
  \STATE \algorithmicif{ $\min(\card{S_u},\card{S_v})=0$} \algorithmicthen{ \algorithmicreturn{  $0$}}
  \STATE Let $k = \ceil{96\eps^{-3}\ln(12/\nu) / \max(\card{S_u},\card{S_v})}$.
  \STATE \algorithmicif{ $k>1$} \algorithmicthen{ replace $S_u$ and $S_v$ by their scaled up versions $S_u\times[k]$ and $S_v\times[k]$.}  \label{step:enlarge}
  \STATE Let $\HFset$ be a representative family of hash functions with parameters $\out=8\max(\card{S_u},\card{S_v})/\eps$, $\beta = \eps/4, \alpha = \eps^2/8$. Let $\famsize = \card{\HFset}$ be its size, and let us index its elements: $\HFset = (h_i)_{i\in[\famsize]}$. \label{step:params}
  \STATE $u$ and $v$ jointly pick a random number $i_{e} \in [\famsize]$, use $h=h_{i_e}\in \HFset^{\out}$ as shared hash function.
  \STATE $u$ sends $h(T_u)$ to $v$, $v$ sends $h(T_v)$ to $u$, where $T_u=\hit{S_u}{h}{S_u}$ and  $T_v=\hit{S_v}{h}{S_v}$.
  \RETURN $\card{h(T_u) \cap h(T_v)}\out/(\samp\cdot k)$
\end{algorithmic}
\end{algorithm}

\begin{lemma}
\label{lem:estimatesimilarity}
    $\estimatesimilarity(\eps)$ outputs a value within $\eps \max(\card{S_u},\card{S_v})$ of $\card{S_u\cap S_v}$, w.p.\ $1-\nu$. It uses $O(1)$ messages of $O\parens*{
    \eps^{-4}\log(1/\nu) + \log\log\card{\rhfspace} + \log \max(\card{S_u},\card{S_v})
    }$ bits.
\end{lemma}
\begin{proof}
    Step~\ref*{step:enlarge} ensures that the parameters we set in step~\ref*{step:params} satisfy the hypotheses of \cref{lem:representative_hash_functions}, i.e., $\out \in \Omega(\eps^{-4}\ln(1/\nu))$, by making the sets artificially bigger if needed. This is done by replacing the original sets by their Cartesian products with a simple set of size $k$: $S'_u=S_u \times [k]$ and $S'_v=S_v \times [k]$, living in the bigger universe $\rhfspace \times [k]$. Clearly, $\card{S'_u \cap S'_v} = k\cdot\card{S_u \cap S_v}$ and $\max(\card{S'_u},\card{S'_v}) = k\cdot \max(\card{S_u},\card{S_v})$, so if $s$ is an estimate for $\card{S'_u \cap S'_v}$ accurate up to $\eps \max(\card{S'_u},\card{S'_v})$ then $s/k$ is an estimate for $\card{S_u \cap S_v}$ accurate up to $\eps \max(\card{S_u},\card{S_v})$. As $k$ is at most $O(\eps^{-3}\ln(1/\nu))$, messages describing an element from the representative hash function family remain of order $O(\log(1/\eps) + \log(1/\nu) + \log\log\card{\rhfspace} + \log(\max(\card{S_u},\card{S_v}))$ even as we scale up the sets. From now on, we ignore $k$, i.e., assume its value to be $1$.

    With $\out = 8 \max(\card{S_u},\card{S_v})/\eps$, $\beta = \eps/4$, $\alpha = \eps^2/8$, we have that $\card{S_u\cup S_v} \leq 2\max(\card{S_u},\card{S_v}) \leq \beta \out$. 
    
    We first show that the estimate we get is a good approximate lower bound on the intersection, and then show it is a good approximate upper bound.
    
    Suppose $\card{S_u \cap S_v} \geq \alpha \out = \eps \max(\card{S_u},\card{S_v})$.  Then, by \cref{lem:representative_hash_functions}, $\card{\hit{(S_u \cap S_v)}{h}{(S_u \cup S_v)}} \geq (1-\beta)\samp\card{S_u \cap S_v}/\out$ w.p.\ $1-\nu$. Since (by \cref{eq:rhf-inclusion}):
    \begin{align}
        &\hit{(S_u \cap S_v)}{h}{(S_u \cup S_v)} \quad\subseteq\quad \hit{(S_u \cap S_v)}{h}{S_u} \quad\subseteq\quad \hit{S_u}{h}{S_u}\\
        \textrm{and}\qquad 
        &\hit{(S_u \cap S_v)}{h}{(S_u \cup S_v)} \quad\subseteq\quad \hit{(S_u \cap S_v)}{h}{S_v} \quad\subseteq\quad \hit{S_v}{h}{S_v}\\
    \intertext{it holds that}
    & \hit{(S_u \cap S_v)}{h}{(S_u \cup S_v)} \quad \subseteq\quad  (\hit{S_u}{h}{S_u}) \cap (\hit{S_v}{h}{S_v}) \quad = \quad T_u \cap T_v\ ,\nonumber
    \end{align}
    and $\card{h(T_u)\cap h(T_v)} \geq \card{T_u \cap T_v}$, so $\card{h(T_u)\cap h(T_v)} \geq (1-\beta)\samp\card{S_u \cap S_v}/\out$. Hence, when $\card{S_u \cap S_v}\geq \alpha\out$, the estimate we give is at most $\beta \card{S_u \cap S_v} \leq \eps\max(\card{S_u},\card{S_v})/4$ lower than the true value, w.p.\ $1-\nu$. When $\card{S_u \cap S_v}\leq \alpha\out$, the estimate cannot be lower than $\card{S_u \cap S_v} - \eps\max(\card{S_u},\card{S_v}) \leq 0$, since the estimate is always positive, so is within $\eps\max(\card{S_u},\card{S_v})$ of the true value.
    
    For the other direction, notice that there are at most as many elements in the intersection of $h(T_u)$ and $h(T_v)$ as there are elements in $S_u \cap S_v$ plus elements in $S_u \cup S_v$ in collision with another element of $S_u \cup S_v$:
    \[\card*{h(T_u) \cap h(T_v)} \quad\leq\quad \card*{\inj{(S_u \cap S_v)}{h}} \  + \  \card*{\collide{(S_u\cup S_v)}{h}{(S_u\cup S_v)}} \]
    
    When $\card{S_u\cap S_v}\geq \alpha\out$, this gives that $\card{h(T_u) \cap h(T_v)} \leq (1+\beta)\card{S_u\cap S_v}\samp/\out + \beta \card{S_u\cup S_v}\samp/\out$, hence that the estimate overestimates the result by at most $2\beta\card{S_u\cup S_v} \leq 4\beta \max(\card{S_u},\card{S_v}) \leq \eps \max(\card{S_u},\card{S_v})$. The same holds when $\card{S_u\cap S_v}\leq \alpha\out$.
    
    The communication cost is obtained by adding the cost $\samp$ of sending $h(T_u)$ and $h(T_v)$ (as they are subsets of $[\samp]$), and the cost $\log \famsize$ of sending the index of a representative hash function in the family.
\end{proof}

As evoked before, almost the same algorithm can be used by the nodes $u$ and $v$ to jointly sample elements from the intersection of their sets by selecting a random element $y$ in $h(T_u)\cap h(T_v)$ and (respectively) outputting the single element in $S_u \cap h^{-1}(y)$ and $S_v \cap h^{-1}(y)$. Since when $S_u \cap S_v$ is large -- at least $\alpha \out$ -- a large fraction of the elements of $h(T_u)\cap h(T_v)$ are images of elements of $S_u \cap S_v$, $u$ and $v$ are likely to sample elements from $S_u \cap S_v$ this way.

\begin{algorithm}[H]\caption{{\jointsample}$(\eps)$, for edge $e=uv$, with sets $S_u,S_v$}
\label{alg:sample}
  \begin{algorithmic}[1]
  \STATE \algorithmicif{ $\min(\card{S_u},\card{S_v})=0$} \algorithmicthen{ \algorithmicreturn{  $0$}}
  \STATE Let $k = \ceil{96\eps^{-3}\ln(12/\nu) / \max(\card{S_u},\card{S_v})}$.
  \STATE \algorithmicif{ $k>1$} \algorithmicthen{ replace $S_u$ and $S_v$ by their scaled up versions $S_u\times[k]$ and $S_v\times[k]$.}
  \STATE Let $\HFset$ be a representative family of hash functions with parameters $\out=8\max(\card{S_u},\card{S_v})/\eps$, $\beta = \eps/4, \alpha = \eps^2/8$. Let $\famsize = \card{\HFset}$ be its size, and let us index its elements: $\HFset = (h_i)_{i\in[\famsize]}$.
  \STATE $u$ and $v$ jointly pick a random number $i_{e} \in [\famsize]$, use $h=h_{i_e}\in \HFset^{\out}$ as shared hash function.
  \STATE Let $J=\card{h(T_u) \cap h(T_v)}$. \algorithmicif{ $J = 0$ } \algorithmicthen{ \algorithmicreturn{  nothing}}.
  \STATE $u$ and $v$ jointly pick a random number $j_{e} \in [J]$. \label{step:random-image}
  \RETURN $h^{-1}(j_e) \cap T_u$ on $u$'s side, $h^{-1}(j_e) \cap T_v$ on $v$'s side.
\end{algorithmic}
\end{algorithm}

\begin{lemma}
\label{lem:jointsample}
    When $\card{S_u\cap S_v}\geq \eps\max(\card{S_u},\card{S_v})$, two nodes $u$ and $v$ running $\jointsample(\eps)$ output the same element at the end of the algorithm, w.p.\ $1-5\eps/4 -\nu$.
\end{lemma}
\begin{proof}
The result follows naturally from $h(T_u)\cap h(T_v)$ containing at least $(1-\eps/4)\card{S_u\cap S_v}$ elements of $h(S_u\cap S_v)$, and $h(T_u)\cap h(T_v)$ being of size at most $(1+\eps)\card{S_u\cap S_v}$, w.p.\ $1-\nu$, as argued in the proof of \cref{lem:estimatesimilarity}.
\end{proof}

The nodes can even sample multiple elements from the intersection of their sets by picking multiple indices instead of a single one in step~\ref*{step:random-image}. This takes the same number of \CONGEST rounds. The only caveat is that some sampled elements might be duplicates when $k = \ceil{96\eps^{-3}\ln(12/\nu) / \max(\card{S_u},\card{S_v})} >1$.

While \jointsample shows some of the ideas we will use later to try multiple colors in a single round, the fact that it only involves two parties means that the procedure may have been designed in a simpler manner, invoking Newman's theorem~\cite{Newman91}. The way we use representative hash functions later to sample random colors is however very multiparty, and may not be derived from statements about public vs private randomness in 2-party communication complexity.

\subsection{Application: Sparsity}
\label{sec:sparsity}

A number of recent algorithms for distributed coloring and other problems treat nodes differently depending on a measure called \emph{sparsity}. Intuitively, sparsity measures the number of missing edges in a node's neighborhood. Depending on the problem, two definitions of sparsity are in use.

\begin{definition}
\label{def:sparsity}
For any subset of the nodes $S\subseteq V$, let $m(S):= \card{E[S]}$ be the number of edges between nodes of $S$. The \emph{global sparsity} of a node $v$ is defined as:
\[\gspar_v = \frac{1}{\Delta} \parens*{\binom{\Delta}{2} - m(N(v))} = \frac{\Delta-1}{2} - \frac{1}{2 \Delta}\sum_{u\in N(v)} \card{N(u) \cap N(v)}\]
The \emph{local sparsity} of a node $v$ is defined as: 
\[\lspar_v = \frac{1}{d_v} \parens*{\binom{d_v}{2} - m(N(v))} = \frac{d_v-1}{2} - \frac{1}{2 d_v}\sum_{u\in N(v)} \card{N(u) \cap N(v)}\]
\end{definition}

The global sparsity is the definition of sparsity generally used in algorithms that solve a coloring problem in which each node can select its own color from $\Delta+1$ colors. The local sparsity is the definition generally used when nodes have only $\deg+1$ colors to choose from.

\estimatesimilarity immediately gives an efficient way to estimate both definitions of sparsity -- under an assumption for the local sparsity. We give the analysis for the global sparsity, and later explain the caveat with local sparsity.

\begin{algorithm}[H]\caption{{\estimatesparsity}$(\eps)$, for $v$ (for global sparsity)}
\label{alg:sparsity}
  \begin{algorithmic}[1]
  \FORALL{$u \in N(v)$}
  \STATE $v$ runs $\estimatesimilarity(\eps/2)$ with $u$ to get an estimate of $s_u=\card{N(u)\cap N(v)}$.
  \ENDFOR
  \STATE $v$ outputs $\frac{\Delta-1}{2} - \frac{1}{2\Delta}\sum_u s_u$ as estimate for its sparsity. \label{step:output-sparsity}
\end{algorithmic}
\end{algorithm}

\begin{lemma}
\label{lem:estimatesparsity}
    $\estimatesparsity(\eps)$ outputs an estimate of $\gspar_v$ that is $\eps \Delta$-close to the true value, w.p.\ $1-(\nu \Delta)^{\eps \Delta/2}$.
\end{lemma}
\begin{proof}
For the estimate to be off by $\eps \Delta$ or more, at least $\eps \Delta/2$ neighbors of $v$ must give an estimate $s_u$ that is at least $\eps \Delta/2$ off. For a subset of $\eps \Delta/2$ neighbors of $v$, the probability that they all give a bad estimate is at most $\nu^{\eps \Delta/2}$. So the probability that such an all-failing subset exists is at most:
\[\binom{\Delta}{\eps \Delta/2} \nu^{\eps \Delta/2} \leq (\nu \Delta)^{\eps \Delta/2}\ .\qedhere\]
\end{proof}

Note that this means \estimatesparsity works w.h.p.\ when $\nu \in 1/\poly(n)$, as well as when $\nu \in 1/\poly(\Delta)$ and $\eps \Delta \in \Omega(\log n/\log\log n)$.

\paragraph{Estimating local sparsity.}
Local sparsity can be similarly estimated, with a caveat. The accuracy of \estimatesimilarity depends on the sizes of the sets we are dealing with. For global sparsity, i.e., in the $\Delta+1$ setting, the global bound of $\Delta$ on the degrees of all nodes implies that $\estimatesimilarity(\eps)$ gives an estimate of $\card{N(u)\cap N(v)}$ within $\eps \Delta$ of the true value.

The difficulty with local sparsity comes from higher degree neighbors. If we could guarantee that each estimate $s_u$ of $\card{N(u) \cap N(v)}$ for $u \in N(v)$ was accurate up to $\eps d_v$ w.p.\ $1-\nu$, we would only need to replace $\Delta$ by $d_v$ in the formula at the end of \cref{alg:sparsity}, i.e., output $\frac{d_v-1}{2} - \frac{1}{2d_v}\sum_u s_u$, to get an estimate of the local sparsity within $\eps d_v$ w.p.\ $(\nu d_v)^{\eps d_v/2}$. Unfortunately, $\estimatesimilarity(\eps)$ only gives an estimate of $\card{N(u)\cap N(v)}$ within $\eps\max(d_u,d_v)$ of the true value, which might be completely off if $d_u \gg d_v$, e.g., with $d_u \geq d_v / \eps$. We hence only claim that we are able to estimate the local sparsity of nodes which do not have too many neighbors of much higher degree.

\begin{lemma}
    Let a node $v$ have less than $\eps d_v/3$ neighbors of degree $\geq 2d_v$. $\estimatesparsity(\eps)$ can be tweaked to output an estimate of $\lspar_v$ that is $\eps d_v$-close to the true value, w.p.\ $1-(\nu d_v)^{\eps d_v/3}$.
\end{lemma}
\begin{proof}
$\eps d_v/3$ nodes can contribute at most $\eps d_v/3$ to the sparsity, so estimating the number of missing edges within the rest of $v$'s neighborhood with precision $2\eps d_v/3$ suffices to get a $\eps d_v$-accurate estimate of $v$'s local sparsity. Since the rest of $v$'s neighborhood has degree at most $2 d_v$, we can run $\estimatesparsity(\eps/3)$ on the subgraph it induces to get an $2\eps d_v/3$ estimate of $v$'s sparsity in this subgraph, giving the result.
\end{proof}

\subsection{Application: Local Triangle Finding}
\label{sec:triangles}

In standard property testing, the goal is to detect with constant probability if a graph is far from satisfying a property. For example, the task may be to distinguish with constant probability whether a graph contains no triangle vs whether an $\eps$-fraction of the edges needs to be deleted for the graph to contain no triangle. The task is solved distributedly but is global in several ways: the notion of distance between graphs takes into account the whole graph and the goal is only for one node of the graph to detect the property. Our \estimatesimilarity primitive allows us to solve a related but more local task: make every edge involved in many triangles detect that it is so.

\begin{theorem}
\label{thm:triangle-finding}
There exists an $O(\eps^{-4})$-round randomized \CONGEST algorithm that, for each edge, detects w.h.p.\ when it is part of $\eps \Delta$ triangles.
\end{theorem}
\begin{proof}
On each edge $uv$, estimate the size of the intersection $\card{N(u)\cap N(v)}$.
\end{proof}

Compared to the usual property testing setting, we solve a harder problem in that we solve the problem with high probability instead of constant probability, and solve it on each edge instead of globally. However our condition for detection is incomparable with that of the property testing setting: our algorithm works whenever a single edge is part of $\eps \Delta$ triangles, while property testing typically assumes $\eps$-farness, i.e., that $\eps\card{E}$ edges have to be deleted from the graph to make it triangle-free. The two are incomparable, since our condition being satisfied on some edge only implies that the graph is $(\eps \Delta/\card{E})$-far from being triangle-free, while in an $\eps$-far graph, it can be the case that each edge is only part of at most one triangle while $\Delta \in \Theta(n)$.

\subsection{Application: Local 4-Cycle Finding}
\label{sec:cycles}

Our technique also allows us to detect 4-cycle locally in \CONGEST, with the same tradeoff as in the detection of triangles compared to standard property testing. Our result is stronger in that we detect occurrence of the pattern locally instead of globally, and we have a higher success probability, but on the other end the two settings are incomparable in that an $\eps$-far graph might not satisfy our local threshold for detection anywhere in any meaningful sense and vice-versa.

\begin{theorem}
\label{thm:cycle-detection}
There exists an $O(\eps^{-4})$-round \CONGEST algorithm that, for each pair of edges incident on the same vertex, detects w.h.p.\ when they are part of $\eps \Delta$ 4-cycles.
\end{theorem}
\begin{proof}
Let $v$ be a vertex of the graph. $v$ picks a random representative hash function $h$ and sends it to all its neighbors $u\in N(v)$, who answer with $\hit[\leq \samp]{N(u)}{h}{N(u)}$. For each pair of neighbors $u,u'$ of $v$, $v$ then estimates $\card{N(u) \cap N(u')}$ with these hashes, as is done in \estimatesimilarity.
\end{proof}

\section{Ultrafast Coloring in Congest}

The techniques we presented allow us to implement all steps of a recent \degoLC algorithm in \CONGEST. The correctness of the algorithm is in~\cite{HKNT21}. We give a succinct but complete description of the algorithm in \cref{app:deg+1_alg} for reference.

\congestDegolcTheorem*

The algorithm as a whole is bandwidth-efficient, but a larger bandwidth is assumed in four places. Two particularly stand out, and are the focus of the upcoming \cref{sec:multitrial,sec:acd}. We sketch how to adapt the rest of the algorithm in \cref{sec:minor}, with the rigorous treatment of these last minor modifications deferred to \cref{app:details}.

The most challenging step to implement in \CONGEST is a method for sampling and ``trying" a set of colors, called ``MultiTrial". We detail its implementation in \cref{sec:multitrial}. Another non-trivial step of computing an almost-clique decomposition is dealt with in \cref{sec:acd}.

\subsection{MultiTrial}
\label{sec:multitrial}

When breaking down recent ultrafast ($O(\log^*n)$ rounds for graphs with large enough degrees) algorithms, all of them contain a step that stands out in the amount of information it requires. Intuitively, in those algorithms, some nodes try up to $\Theta(\log n)$ colors over the course of the algorithm, with the idea that if each color succeeds with constant probability, then the nodes get colored w.h.p.\ by trying that many colors. However, $\log n$ arbitrary colors take at least $\Theta(\log \Delta \log n)$ bits to describe (and possibly much more when nodes are given lists of colors instead of using $[\deg+1]$ or $[\Delta+1]$), which would require $\Omega(\log \Delta)$ rounds in \CONGEST. A more communication-efficient procedure following the same idea needs to compromise on some front, which we do here by compromising on the randomness and accuracy of the colors that nodes try, using representative hash functions.

We give a procedure -- {\multitrial} -- that within bandwidth $\bw$ allows a node $v$ to try $x$ random colors from its palette, where $x$ can be as large as $\Theta(\bw)$.
While trying $\Theta(b)$ colors in a single round is straightforward in \LOCAL, a na\"ive implementation in \CONGEST would take $\Omega(\log \card{\colSpace})$ rounds for a color space $\colSpace$.
We achieve similar results in \CONGEST by replacing the random sampling of colors by a pseudorandom one.
While the $x$ colors tried are not as random as $x$ independent random samples, enough randomness is used so that one of those colors succeeds w.p.\ $1-\exp(-\Omega(x))-\exp(-\Omega(\bw))$, essentially the same probability as if they were independent. With bandwidth $\bw \in \Theta(\log n)$, this allows up to $\Theta(\log n)$ colors to be tried in a single round, and for a node to be colored with probability $1-1/\poly(n)$.
Previously, this was only known to be possible in the very restricted setting of locally sparse graphs~\cite{HN21}.

To get an intuitive understanding of our approach, let us assume that each node $v$ can sample and communicate to its neighbors a random hash function $h_v : \colSpace \rightarrow [\out]=\set{1,\ldots,\out}$ for a number $\out$ of its choice. To have all nodes try $x$ colors, on each edge $uv$, node $v$ sends to $u$ the hash values of the color it tries through $h_u$ (and vice versa). If $v$ tries a color $\col$ that hashes to a value different from all the hash values it received, $v$ can safely color itself with $\col$. To make the procedure more efficient, we have 
$v$ pick random colors among those with a hash value $\leq \samp = O(\log n)$ through $h_v$. With this restriction, the neighbors of $v$ only need to tell $v$ about the colors they try that hash to a value $\leq \samp$ through $h_v$. 
This uses $\samp = O(\log n)$ bits of communication using a $\samp$-size bitstring.

For this to work, the hash function must satisfy three properties: first, enough colors must hash to a value $\leq \samp = O(\log n)$; second, collisions must be rare enough for a unique hash to be sampled; and third, it should be possible to communicate a hash function in $O(\log n)$ bits so the process takes $O(1)$ rounds. Increasing $\out$ reduces the number of collisions, but reduces how many elements hash to a value $\leq \samp = O(\log n)$, so a balance must be found. This balance is found at $\out \in \Theta(\card{\pal_v})$.

Using families of representative hash functions, whose existence we proved in \cref{lem:representative_hash_functions}, we show how to implement \multitrial efficiently in \CONGEST (\cref{alg:multitrial} and \cref{lem:multitrial-success}).

The pseudocode of {\multitrial} is presented in \cref{alg:multitrial}.
Let $\alpha=1/12$, $\beta=1/3$, and for each $\out\in \bbN$, let $\nu_\out = \max(n^{-c},12\exp(-\alpha\out/45))$ and $\samp_\out \in \Theta\left(\beta^{-2}\alpha^{-1}\log(1/\nu_\out)\right)$, for a constant $c>3$ (hence, even for $n^2$ events of probability $\nu_\out$, when $\out \in \omega(\log n)$, none occurs w.h.p.). We assume that all the nodes know, for each $\out \in [2\beta^{-1}\Delta]=[6\Delta]$, a common family of hash functions $\HFset^\out=\parens{h^{(\out)}_i}_{i\in[\famsize]}\subseteq [\out]^{\colSpace}$ and value $\samp_\out$ with the properties of \cref{lem:representative_hash_functions}. This could be achieved, e.g., by having each node compute the lexicographically first such pair of family and parameter, for each $\out$. Note that $\samp_\out \in O(\log n)$ for all $\out$, and that this parameter can be chosen to be the same $\samp=\Theta(\log n)$ for all values of $\out \in \omega(\log n)$.

\begin{algorithm}[H]\caption{{\multitrial}($x$), for node $v$}
\label{alg:multitrial}
  \begin{algorithmic}[1]
    \STATE Let $\out_v\gets 6\card*{\pal_v}$, pick a random $h_v=h^{(\out_v)}_{i_v}\in \HFset^{\out_v}$, broadcast $\out_v,i_v$ to $N(v)$. 
    \STATE $X_v\gets$ $x$ independently chosen random colors in $\hit{\pal_v}{h_v}{\pal_v}$.\label{st:mulxv}
    \FORALL {$u\in N(v)$ and all $i\in [\samp_{\out_u}]$}
    \STATE \algorithmicif\  $\exists \col \in X_v$, $h_u(\col)=i$ \algorithmicthen\  $b_{v \rightarrow u}[i]\gets 1$\  \algorithmicelse\  $b_{v \rightarrow u}[i]\gets 0$ \ENDFOR
    \STATE Send $b_{v \rightarrow u}$ and receive $b_{u \rightarrow v}$  to/from $u$, for all $u\in N(v)$.
    \IF {$\exists \col \in X_v$ s.t. $\forall u\in N(v)$, $b_{u \rightarrow v}[h_v(\col)]=0$} 
    \STATE Adopt some such $\col$ as permanent color and broadcast to $N(v)$.
    \ENDIF
\end{algorithmic}
\end{algorithm}

\begin{lemma}
\label{lem:multitrial-success}
For every node $v$, if $x \leq \card{\pal_v}/2\card{N(v)}$, then  
an execution of {\multitrial}$(x)$ colors $v$ with probability $1-(7/8)^{x}-2\nu$, where $\nu \leq e^{-\Theta( \card{\pal_v})} + n^{-\Theta(1)}$, even when conditioned on any particular combination of random choices of the other nodes.
\end{lemma}

\begin{proof}
    Consider $Y_v=\bigcup_{u\in N(v)} X_u$, the set of colors tried by neighbors of $v$.
    Note that $|Y_v|\le x|N(v)|\le |\pal_v|/2\le \out_v/12$ (recall $\out_v=6 \card*{\pal_v}$), and its composition is independent from $v$'s choice of random colors. Letting  $T_v=\pal_v \setminus Y_v$ and $P_v = \pal_v \cup Y_v$, we have $\card{P_v},\card{T_v},\card{\pal_v} \in [\out_v/12,\out_v/3]$, and so, the triplets $(\out_v,P_v,T_v)$ and $(\out_v,P_v,\pal_v)$ satisfy  \cref{lem:representative_hash_functions} with our parameters $\alpha,\beta,\nu$. Let $\samp = \samp_{\out_v}$. 
    The lemma implies that w.p.\ $1-\nu$,  $\card*{\hit{\pal_v}{h_v}{\pal_v}} \le (1+\beta)\cdot \samp\card{\pal_v}/\out_v \le 2\samp/9$,  
    and similarly, w.p.\ $1-\nu$, $\card*{\hit{T_v}{h_v}{P_v}} \ge (1-2\beta)\cdot \samp\card{T_v}/\out_v \ge \samp/36$. 
    Since additionally $(\hit{T_v}{h_v}{P_v}) \subseteq (\hit{\pal_v}{h_v}{P_v}) \subseteq (\hit{\pal_v}{h_v}{\pal_v})$, we conclude that $\hit{T_v}{h_v}{P_v}$ forms a $(\samp/36) / (2\samp/9) = 1/8$ fraction of $\hit{\pal_v}{h_v}{\pal_v}$, and  any color randomly picked in $\hit{\pal_v}{h_v}{\pal_v}$ is  in $\hit{T_v}{h_v}{P_v}$ w.p.\ at least $1/8$.
    Hence, conditioned on the $1-2\nu$ probability event that $|\hit{T_v}{h_v}{P_v}|\ge |\hit{\pal_v}{h_v}{\pal_v}|/8$,  
    the $x$ colors randomly picked by $v$ in  $\hit{\pal_v}{h_v}{\pal_v}$ all miss $\hit{T_v}{h_v}{P_v}$ w.p.\ at most $(7/8)^{x}$. As any color found in $\hit{T_v}{h_v}{P_v}$ will be successful for $v$, $v$ gets colored w.p.\ $1-(7/8)^{x}$, conditioned on an event of probability $1-2\nu$. 
\end{proof}

\subsection{Almost-Clique Decomposition}
\label{sec:acd}

Almost-clique decompositions are commonly defined and computed according to a relation that classifies two connected nodes as \emph{friends} if they share most of their neighborhoods. Nodes whose neighborhood is almost all friends have dense neighborhoods, i.e., most pairs of nodes in their neighborhood are connected by an edge, and are mostly adjacent to nodes of similar degree. Reciprocally, nodes with few friends are either \emph{uneven}, i.e., adjacent to many nodes of much higher degree, or have a sparse neighborhood, i.e., a large fraction of their neighbors are not directly connected.

\begin{definition}[\cite{AA20}]
    Let $\eps \in [0,1]$. An edge $uv$ is
    \begin{itemize}
    \item \emph{$\eps$-balanced} iff $\min(d_u,d_v) \geq (1-\eps)\max(d_u,d_v)$,
    \item \emph{$\eps$-friend} iff it is $\eps$-balanced and $\card{N(u)\cap N(v)} \geq (1-\eps)\min(d_u,d_v)$.
    \end{itemize}
\end{definition}

Computing $\eps$-\friend predicates exactly would be too costly in many models of computation where almost-clique decompositions are relevant. Fortunately, this much accuracy is not needed: in \CONGEST, it suffices to have access to a procedure $\eps$-\buddy that distinguishes between an edge being $\eps$-friend and it being far from it, i.e., not $c\cdot \eps$-friends for some constant $c>1$ (see, e.g., Appendix~B in~\cite{HKMN20} for details). The idea of computing almost-clique decompositions using a sampling-based approach originated in~\cite{ACK19}.
This can be done easily by testing whether $d_u$ and $d_v$ are approximately equal and then running \estimatesimilarity if this is the case, which works w.h.p.\ with bandwidth $\log n$.

\subsection{Final Minor Modifications}
\label{sec:minor}

While the bulk of adapting the \LOCAL algorithm of~\cite{HKNT21} for $\degoLC$ to \CONGEST is figuring out how to efficiently try up to $\Theta(\log n)$ colors (\multitrial) and compute an almost-clique decomposition (\computeacd) within the $O(\log n)$ bandwidth constraint, a few additional minor modifications are required. We defer their rigorous treatment to \cref{app:details}, and sketch here the essence of those changes.

\paragraph{Leader selection.} In the original \LOCAL algorithm, a node is chosen as leader in each almost-clique based on a quantity called \emph{slackability} (see \cref{app:acd} for its definition). As the slackability of a node is entirely determined by its neighborhood (palettes included), finding the node of minimum slackability within each almost-clique is trivial in \LOCAL, and only takes $O(1)$ rounds. This is no longer the case in \CONGEST. This is circumvented by arguing that it suffices to identify a node of low but not necessarily minimal slackability within each almost-clique, and that the slackability can be estimated with the needed accuracy efficiently in \CONGEST. The details of these two arguments are given in \cref{sec:leader}.

\paragraph{Coloring put-aside sets.} In the \degoLC algorithm we are adapting to \CONGEST, almost-cliques and their nodes are dealt with differently depending on whether their sparsity is above or below some threshold. Very dense almost-cliques are dealt with by putting aside a subset of its nodes to color later, in order to provide temporary slack to the rest of the almost-cliques. To color those nodes at the end of the algorithm, information about their palettes and how they are connected is centralized. How all this information can be centralized in \CONGEST is not as simple as in \LOCAL, and requires in particular more control on the sizes of the put-aside sets. The adaptation of this part of the algorithm is detailed in \cref{sec:coloring-put-aside}.

\paragraph{Handling large colors.} An important aspect of list-coloring problems in models with a bandwidth constraint is that colors may live in a color space bigger than $2^{\Theta(\bw)}$, i.e., too big for the nodes to send a color in a constant number of rounds. Our procedure \multitrial circumvents this, in the case of trying multiple colors, by hashing. This is fortunately possible in other parts of the algorithm, e.g., whenever nodes need to inform their neighbors of their newly adopted color, or need to inform another node of the color it should try. We show that color spaces of size up to $\exp(n^{\Theta(1)})$ can be handled without increasing the complexity of the algorithm. How this is done is explained in \cref{sec:large-colors}.

\section{Uniform Implementation}
\label{sec:uniform}

\Cref{lem:representative_hash_functions} -- on the existence of representative hash functions -- does not give an explicit construction. Hence, an algorithm relying on their existence either needs to assume that the nodes receive a common family of such hash functions ``for free'' at the beginning, or have the nodes find a common family of representative hash functions themselves. This second option requires extensive computational resources, as it involves exploring the space of $F$-element subsets of $\range{\out}^\colspace$, performing expensive statistical tests on each subset.

In this section, we show how our subroutines that use representative hash functions can be modified to not rely on them. With these new implementations, the nodes only have to perform computations polynomial in $n$ and $\Delta$ in our algorithm.
We leave as an open question the explicit construction of families of representative hash functions.

In our new uniform implementations of \multitrial and \buddy, a key idea is to introduce and exploit some asymmetry between the parties. By having one of the parties partially \emph{choose} a hash function instead of taking it fully at random, this party can ensure that not too many collisions occur between the elements it knows of. To remove the reliance on representative hash functions, we make use of other objects with explicit constructions: pairwise-independent hash functions, representative multisets (constructed from averaging samplers), and error-correcting codes. Subroutines other than \multitrial and \buddy do not rely on representative hash functions, and as such do not need to be modified for the algorithm to be uniform.

\subsection{MultiTrial}

The core properties of \multitrial are twofold: first, \multitrial is able to describe up to $\Theta(b)$ colors in a single $b$-sized \CONGEST message; second, those colors are sufficiently random that it is as if they each had a constant probability of success, i.e., when trying $k$ colors, at least one succeeds w.p.\ $1-\exp(-\Omega(k))$.

Let $\HFpwi^{\out}$ be a set of $\eps$-almost pairwise-independent hash functions from $\colspace$ to $\out$ (see, e.g., Problem~3.4 in \cite{Vadhan12}). When selecting a hash function from such a family, \[\Pr_{h \overset{R}{\gets} \HFpwi^{\out}}[h(x_1)=y_1 \textrm{ and } h(x_2)=y_2] \leq \frac{1+\eps}{\out^2}, \qquad\forall x_1,x_2 \in \colspace, y_1,y_2\in[\out] \textrm{ with } x_1\neq x_2~.\]

There exist explicit such families of size $\poly(\out ,\log \card{C}, 1/\eps)$, such that sampling an element from the family only requires to pick $(\log \out + \log \log \card{C} + \log( 1/\eps))$ random bits.

We sketch the argument showing that a \multitrial procedure that only relies on explicit constructs (and thus, does not rely on representative hash functions at the moment) is possible. We give pseudocode of this procedure below. Its core idea is to have each node $v$ select a hash function that has few collisions in its palette. As a result, the image of $v$'s palette through the hash function it picked is almost of the same size as the palette itself. Let us now consider the images of the colors tried by $v$'s neighbor through this hash function. If the total number of colors tried by the neighbors of $v$ is less than half the number of colors in $v$'s palette, then a constant fraction of the image of $v$'s palette is necessarily not the image of any color tried by a neighbor of $v$. Therefore, by sampling $x$ hashes from the image of its palette using an explicit representative multiset over the space of hashes (\cite{HN21}, and \cref{app:representative-multisets}), $v$ can succeed in securing a random color w.p.\ $1-\exp(-\Omega(x))-\exp(-\Omega(\bw))$.

\begin{algorithm}[H]\caption{Uniform {\multitrial}($x$), for node $v$}
\label{alg:unif_multitrial}
  \begin{algorithmic}[1]
    \STATE Let $\out_v\gets 6\card*{\pal_v}$, pick a random $h_v=h^{(\out_v)}_{i_v}\in \HFpwi^{\out_v}$ with at most $\out_v/3$ collisions in $\pal_v$, broadcast $\out_v,i_v$ to $N(v)$. 
    \STATE Let $\samp_v = \min(\bw,\out_v)$, pick a random representative multiset $S_v$ of size $\samp_v$, subset of $\range{\out_v}$. Let $S_v =  \set{s_1^{(v)},\ldots s_{\samp_v}^{(v)}}$.\label{st:repsetxv}
    \STATE $X_v \gets$ $x$ random elements of $\pal_v \cap h_v^{-1}(S_v)$, picked uniformly at random.
    \FORALL {$u\in N(v)$ and all $i\in \range{\samp_u}$}
    \STATE \algorithmicif\  $\exists \col \in X_v$, $h_u(\col)=s_i^{(u)}$ \algorithmicthen\  $b_{v \rightarrow u}[i]\gets 1$
     \algorithmicelse\  $b_{v \rightarrow u}[i]\gets 0$ \ENDFOR
    \STATE Send $b_{v \rightarrow u}$ and receive $b_{u \rightarrow v}$  to/from $u$, for all $u\in N(v)$.
    \IF {$\exists \col \in X_v$ s.t. $\forall u\in N(v)$, $b_{u \rightarrow v}[h_v(\col)]=0$} 
    \STATE Adopt some such $\col$ as permanent color and broadcast to $N(v)$.
    \ENDIF
\end{algorithmic}
\end{algorithm}

\subsection{Almost-Clique Decomposition}

Similar ideas to those that enable a uniform implementation of \multitrial allow for a uniform implementation of $\eps$-\buddy. Again, we provide pseudocode of the procedure below, and sketch the argument for its correctness.

The algorithm takes place between two nodes $u$ and $v$. They first test whether their degrees differ significantly. If they do, the algorithm stops: the edge is not $\eps$-\buddy. Otherwise, one of the nodes chooses an almost pairwise-independent hash function with few collisions between the IDs of its neighborhood. Then, the nodes pick a random representative multiset over the space of hashes, and compute the sampled hashes that are the image of a single of their neighbors.

If the nodes have few hashes in common, they declare the edge non-$\eps$-\buddy, as having few hashes in common is only likely when the nodes' neighborhoods do not mostly intersect.

When the nodes share a lot of hashes, however, two causes are possible: they either share a large part of their neighborhoods, or the hash function that was picked has many collisions between the two neighborhoods. The rest of the algorithm is devoted to distinguishing the two.

To do so, the nodes apply an error-correcting code to the ID of each of their neighbors. As a result, distinct IDs now differ in a constant fraction of their bits. The nodes then each build a bitstring by concatenating the preimages of the hashes they found to have in common. The two resulting bitstrings are guaranteed to be of small Hamming distance if the nodes genuinely share many neighbors, but must differ in a large fraction of indices if the hashes that the nodes found in common were due to collisions. The nodes sample random indices of these bitstrings using representative multisets, exchange the bits at those indices, estimate the Hamming distance between their bitstrings from those bits, and conclude. 
The idea of using an error-correcting code to increase the Hamming distance between distinct bitstrings has been used previously in communication-focused models, for example \cite{Ambainis96}.

In the following pseudocode, $\mathrm{``\text{\textunderscore}"}$ represents an empty bitstring,  $\enc$ is the encoder of an error correcting code, $x_u.\enc(w)$ is the concatenation of bitstrings $x_u$ and $\enc(w)$. The error correcting code is chosen to have parameter, e.g., $[3b,b,b/2]$, where $b\in \Theta(\log n)$ is the number of bits used to write IDs in the graph, meaning that IDs initially written on $b$ get expanded to $3b$ bits, and that two distinct IDs differ by at least $b/2$ bits after the encoding.

\begin{algorithm}[H]\caption{Uniform {$\eps$-\buddy}, for edge $uv$}
\label{alg:unif_buddy}
  \begin{algorithmic}[1]
    \STATE \algorithmicif{ $d_u > d_v/(1-\eps)$} \algorithmicor { $d_v > d_u/(1-\eps)$} \algorithmicthen{ \algorithmicreturn{ \NO}}
    \STATE Let $\out\gets 6\max(d_u,d_v)/\eps$. $v$ chooses an hash function $h=h^{(\out)}_{i}\in \HFpwi^{\out}$ such that at most $\eps d_v/3$ elements in $N(v)$ are involved in a collision, and sends $(\out,i)$ to $u$.
    \STATE Let $\samp = \min(\bw,\out)$, $u$ and $v$ pick a random representative multiset $S$ of size $\samp$, subset of $\range{\out}$. Let $S =  \set{s_1,\ldots, s_{\samp}}$.
    \FORALL {$i\in \range{\samp}$}
    \STATE \algorithmicif\  $\exists ! w \in N(u)$, $h(w)=s_i$ \algorithmicthen\  $b_{u}[i]\gets 1$ \algorithmicelse\  $b_{u}[i]\gets 0$
    \STATE \algorithmicif\  $\exists ! w \in N(v)$, $h(w)=s_i$ \algorithmicthen\  $b_{v}[i]\gets 1$ \algorithmicelse\  $b_{v}[i]\gets 0$ \ENDFOR
    \STATE $u$ and $v$ exchange $b_u$ and $b_v$
    \STATE \algorithmicif{ $\card{\set{i:b_u[i]\cdot b_v[i] = 1}} \leq (1-3\eps)\samp$} \algorithmicthen{ \algorithmicreturn{ \NO}}
    \STATE $x_u \gets \mathrm{``\text{\textunderscore}"}$, $x_v \gets \mathrm{``\text{\textunderscore}"}$
    \FORALL {$i\in \range{\samp}$ s.t.\ $b_u[i] \cdot b_v[i] =1$}
    \STATE $x_u \gets x_u . \enc(w)$ where $w$ is the unique $w \in N(u)$ s.t.\ $h(w)=s_i$
    \STATE $x_v \gets x_v . \enc(w)$ where $w$ is the unique $w \in N(v)$ s.t.\ $h(w)=s_i$
    \ENDFOR
    \STATE Let $\ell = \mathrm{length}(x_u)$, $\samp' = \min(\bw,\ell)$, $u$ and $v$ pick a random representative multiset $S$ of size $\samp'$, subset of $\range{\ell}$.
    \STATE \algorithmicif{  $\card{\set{i\in[\samp']:x_u[i] \neq x_v[i]}} \geq \eps \samp'$ } \algorithmicthen{ \algorithmicreturn{ \NO}} \algorithmicelse{ \algorithmicreturn{ \YES}}
\end{algorithmic}
\end{algorithm}

\section*{Acknowledgements}
This project was supported by the European Union’s Horizon 2020 Research and Innovation Programme under grant agreement no.\ 755839 and by Icelandic Research Fund grants no.\ 174484 and 217965. Part of the work was done while T.~Tonoyan was with the CS Department of the Technion, Israel.


\bibliographystyle{alpha}
\bibliography{references}

\appendix

\section{Concentration Bounds}
\label{app:concentration}

\begin{lemma}[Chernoff bounds]\label{lem:basicchernoff}
Let $\{X_i\}_{i=1}^r$ be a family of independent binary random variables with $\Pr[X_i=1]=q_i$, and let $X=\sum_{i=1}^r X_i$. For any $\delta>0$, $\Pr[|X-\Exp[X]|\ge \delta\Exp[X]]\le 2\exp(-\min(\delta,\delta^2) \Exp[X]/3)$.
\end{lemma}

\begin{lemma}[Hoeffding's inequality \cite{Hoeffding}]
\label{lem:chernoff-various-ranges}
Let $X_1 \ldots X_n$ be $n$ independent random variables distributed
in $\range{a_i,b_i}$, $X := \sum_{i=1}^n X_i$ their sum. For $t >0$:
\[
\Pr\event{\abs{X-\Exp[X]} > t} \leq 2\exp \parens*{ - \frac {2 \cdot t^2} {\sum_i (b_i - a_i)^2} }~.\]
\end{lemma}

We use the following variants of Chernoff bounds for dependent random variables. The first one is obtained, e.g., as a corollary of Lemma 1.8.7 and Thms.\ 1.10.1 and 1.10.5 in~\cite{Doerr2020}.

\begin{lemma}[Martingales \cite{Doerr2020}]\label{lem:chernoff}
Let $\{X_i\}_{i=1}^r$ be binary random variables, and $X=\sum_i X_i$.
    If $\Pr[X_i=1\mid X_1=x_1,\dots,X_{i-1}=x_{i-1}]\le q_i\le 1$, for all $i\in [r]$ and $x_1,\dots,x_{i-1}\in \{0,1\}$ with $\Pr[X_1=x_1,\dots,X_r=x_{i-1}]>0$, then for any $\delta>0$,
    \begin{align*}
    \Pr\event{X\ge(1+\delta)\sum_{i=1}^r q_i}
    &\le \exp\parens*{-\frac{\min(\delta,\delta^2)}{3}\sum_{i=1}^r q_i}~.
    \intertext{
    If $\Pr[X_i=1\mid X_1=x_1,\dots,X_{i-1}=x_{i-1}]\ge q_i$, $q_i\in (0,1)$, for all $i\in [r]$ and $x_1,\dots,x_{i-1}\in \{0,1\}$ with $\Pr[X_1=x_1,\dots,X_r=x_{i-1}]>0$, then for any $\delta\in [0,1]$,
    }
    \Pr[X\le(1-\delta)\sum_{i=1}^r q_i]
    &\le \exp\left(-\frac{\delta^2}{2}\sum_{i=1}^r q_i\right)\ .
    \end{align*}
\end{lemma}

A function $f(x_1,\ldots,x_n)$ is  \emph{$c$-Lipschitz} iff changing any single $x_i$ affects the value of $f$ by at most $c$, and $f$ is  \emph{$r$-certifiable} iff whenever $f(x_1,\ldots,x_n) \geq s$ for some value $s$, there exist $r\cdot s$ inputs $x_{i_1},\ldots,x_{i_{r\cdot s}}$ such that knowing the values of these inputs certifies $f\geq s$ (i.e., $f\geq s$ whatever the values of $x_i$ for $i\not \in \{i_1,\ldots,i_{r\cdot s}\}$).
\begin{lemma}[Talagrand's inequality~\cite{DP09}]
\label{lem:talagrand}
Let $\{X_i\}_{i=1}^n$ be $n$ independent random variables and $f(X_1,\ldots,X_n)$ be a $c$-Lipschitz $r$-certifiable function; then for $t\geq 1$,
\[\Pr\event*{\abs*{f-\Exp[f]}>t+30c\sqrt{r\cdot\Exp[f]}}\leq 4 \cdot \exp\parens*{-\frac{t^2}{8c^2r\Exp[f]}}~.\]
\end{lemma}

\begin{lemma}[Lemma 24 in \cite{HKNT21}]
\label{lem:talagrand-difference}
Let $\set*{X_i}_{i=1}^n$ be $n$ independent random variables. Let $\set*{A_j}_{j=1}^k$ and $\set*{B_j}_{j=1}^k$ be two families of
events that are functions of the $X_i$'s. Let $f=\sum_{j\in[k]} \mathbb{I}_{A_j}$, $g=\sum_{j\in[k]} \mathbb{I}_{A_j \cap \overline{B}_j}$,\footnote{$\mathbb{I}$ denotes the indicator random variable of an event.} and $h=f-g$ be such that $f$ and $g$ are $c$-Lipschitz and $r$-certifiable w.r.t.\ the $X_i$'s, and $\Exp[h] \geq \alpha \Exp[f]$ for some constant $\alpha \in (0,1)$. Let $\delta \in (0,1)$. Then for $\Exp[h]$ large enough:
\[\Pr\event*{\abs*{h - \Exp[h]} > \delta \Exp[h]} \leq \exp(-\Omega(\Exp[h]))~.\]
\end{lemma}

\section{Explicit Representative Multisets}
\label{app:representative-multisets}

For completeness, we give an explicit construction of representative multisets in this section.

Intuitively, a \emph{sampler} for a domain of size $M$ is a function that takes some number $N$ of perfect random bits as input and outputs $t$ elements $z_1,\ldots,z_t$ in $\range{M}$. A \emph{hitting sampler} gives the guarantee that the sampled outputs hit any large enough subset of $\range{M}$ with some probability, while an \emph{averaging sampler} gives the guarantee that for any function $f$ with output in $[0,1]$, the average value of $f$ on the sampled elements is close to the average value of $f$ over the whole domain $\range{M}$.

\begin{definition}[Averaging Samplers]
A function $\Sampler:\range{N} \rightarrow\range{M}^t$ is a $(\delta,\eps)$-\emph{averaging sampler} if for every function $f: \range{M} \rightarrow \range{0,1}$, we have:
\[\Pr_{(z_i)^t_{i=1} \gets \Sampler(U_{\range{N}})}\event*{\frac 1 t \sum_{i=1}^t f(z_i) - \frac{1}{M} \sum_{x\in \range{M}} f(x) > \eps} \leq \delta~.\]
\end{definition}

Averaging samplers are relevant to our setting in the following way: let a node $v$ of palette $\pal_v$ in a color space $\colspace$ of size $\in \range{\Delta, 2\Delta}$ have slack at least $4\eps \card{\colspace}$ for some constant $\eps \in (0,1)$. Let its uncolored neighbors try a total of at most $2\eps \card{\colspace}$ colors in a given round. Let $S$ be the set of colors in $\pal_v$ that is not tried by any of its neighbors, of size at least $2\eps \card{\colspace}$. Let $f$ be the indicator function for $S$. Then, if we use a $(\delta,\eps)$-averaging sampler over $\colspace$ to sample $t$ elements, then with probability at least $1-\delta$, at least $\eps t$ of the sampled elements are in $S$. This means that $v$ trying any of the sampled colors succeeds with probability at least $1-\eps$, conditioned on an event of probability $1-\delta$. If we furthermore assume that the palette of $v$ has size comparable to its slack, e.g., at most $8\eps\card{\colspace}$, and condition on the event that the random sampler does not over-sample elements from $\pal_v$ by more than $\eps$, then having $v$ try a random sampled color that is in its palette succeeds with constant $\Omega(1)$ probability (i.e., independent of $\eps$).

Used that way, a averaging sampler replaces representative sets in most use cases. It may be interpreted as being a family of multisets, by considering the output of an averaging sampler on all possible choices of random bits. Taking $\delta = 1/\poly(n)$ and $\eps\in \Theta(1)$, there exists explicit averaging samplers that use $N=\Theta(\log n)$ random bits as input and sample $t=\Theta(\log \card{\colspace} + \log n)$ elements.

\section{Definitions Related to Almost-Clique Decompositions}
\label{app:acd}

\begin{definition}[Sparsity]
\label{def:sparsity-alt}
The \emph{(local) sparsity} $\zeta_v$ of node $v$ is defined as $\frac{1}{d_v}\cdot\left[\binom{d_v}{2}-m(N(v))\right]$. Node $v$ is \emph{$\zeta$-sparse} if $\zeta_v\ge \zeta$, and \emph{$\zeta$-dense} if $\zeta_v\le \zeta$.
\end{definition}

\begin{definition}[Disparity, Discrepancy \& Unevenness]
\label{def:unevennes}
The \emph{disparity} of $u$ towards $v$ is defined as $\disc_{u,v} = \card{\pal_u \setminus \pal_v} / \card{\pal_u}$. The \emph{discrepancy} of node $v$ is defined as $\disc_v=\sum_{u \in N(v)} \disc_{u,v}$, and
its \emph{unevenness} is defined as $\unev_v=\sum_{u \in N(v)} \frac{\max(0,d_u - d_v)}{d_u+1}$. Node $v$ is \emph{$\disc$-discrepant} if $\disc_v\ge \disc$, \emph{$\unev$-uneven} if $\unev_v\ge \unev$.
\end{definition}

\begin{definition}[($\deg+1$) ACD \cite{AA20}] \label{def:acd}
Let  $G=(V,E)$ be a graph and $\eacd,\espa\in (0,1)$ be parameters. A partition $V=\Vsp \sqcup \Vun \sqcup \Vdense$ of $V$, with $\Vdense$ further partitioned into $\Vdense = \bigsqcup_{C \in \acset} C$,
is an \emph{almost-clique decomposition (ACD)} for $G$ if: \begin{compactenum}
    \item Every $v \in \Vsp$ is $\espa d_v$-sparse\ ,
    \item Every $v \in \Vun$ is $\espa d_v$-uneven\ ,
    \item For every $C \in \acset$ and $v\in C$, $d_v \leq (1+\eacd)\card{C}$\ ,
    \item For every $C \in \acset$ and $v\in C$, $(1+\eacd)|N_{C}(v)|\ge \card{C}$\ .
\end{compactenum}
\end{definition}

The slackability of an almost-clique is defined as $\barsigma_C = \min_{v\in C} \barsigma_v$. As we deal with nodes of degree between $\log^7 \Delta$ and $\Delta$, we set a threshold $\ell = \log^{2.1} \Delta$ and declare almost-clique of lower slackability to be \emph{low-slack}, and almost-clique of higher slackability to be \emph{high-slack}.

\section{Final Details of the Algorithm for (Degree+1)-List-Coloring in CONGEST}
\label{app:details}

We give here the remaining details of the implementation of a $(\deg+1)$-list-coloring algorithm in \CONGEST.
We explain how to find good enough ``leaders" of almost-cliques in \cref{sec:leader}, and compute ``put-aside" sets in \cref{sec:coloring-put-aside}, and finally how to deal with large color values in \cref{sec:large-colors}.
We first give an informal description of the \degoLC algorithm of~\cite{HKNT21}.

\paragraph{Algorithm overview.} The \degoLC algorithm of~\cite{HKNT21} consists of up to $O(\log^* n)$ phases in which all the nodes whose degree falls within a range of the form $\range{\log^7 x, x}$ get colored, w.h.p. Each phase takes at most $\poly(\log \log n)$ rounds to complete, but a phase dealing only with nodes of degree $\log^7 n$ or higher can achieve its task in merely $O(\log^* n)$ rounds, resulting in an $O(\log^* n)$ algorithm when given a graph of minimum degree $\log^7 n$.

Each phase starts with computing an almost-clique decomposition, which partitions the nodes in \emph{uneven}, \emph{sparse}, and \emph{dense} vertices. Dense vertices are themselves partitioned into \emph{almost-cliques}, highly connected (clique-like) subgraphs of diameter at most $2$. See \cref{app:acd} for formal definitions of $\Vun$, $\Vsp$, $\Vde$, and of almost-clique decompositions. Our primitive \estimatesimilarity can be directly used to compute such a decomposition, in a process which we was explained in \cref{sec:acd}.

Within each phase, the algorithm first deals with all the sparse and uneven nodes, then all the dense nodes. In both cases, slack is generated by first having each node try a random color of its palette with some constant probability. The only possible difficulty in implementing this part in \CONGEST is that the nodes may have to communicate colors whose description does not fit in $O(\log n)$ bits. We explain how we can deal with a color space of size up to $\exp(n^{\Theta(1)})$ in \cref{sec:large-colors}. From there, the algorithms differ between the dense and the non-dense case. 

When dealing with sparse and uneven nodes, the algorithm identifies a set $\Vst$, consisting of sparse nodes for which \slackgeneration might not generate permanent slack but which can get temporary slack by being colored early. Instead of identifying this set before slack generation as in the \LOCAL algorithm, which might be hard to do in \CONGEST, we simply let the success of the slack generation process guide our partitioning. More precisely, we let each node $v$ join $\Vst$ if it received less than $\ehat d_v$ permanent slack but is adjacent to at least $\ehat d_v$ uncolored nodes that did. A node that neither received permanent slack nor is adjacent to many nodes that did is added to a set \BAD, which is either empty or shattered due to the probability of success of slack generation.
Indeed, while this process may fail at providing the needed slack to some nodes, \cref{L:slackgen-sparse} (from \cite{HKNT21}) shows that this happens to a node with probability $d^{\omega(1)}$ when we consider nodes of degree in the range $[\log^7 d, d]$. This probability is low enough w.r.t.\ $d$ that all nodes receive the slack they need w.h.p.\ when they are all of degree $\log^7 n$ or more. It is also low enough w.r.t.\ $d$ that when considering nodes of lower degree, the subgraph induced by nodes that do not receive the slack they need is \emph{shattered}, i.e., has $\poly(\log n)$-sized connected components which can be efficiently colored with by a deterministic algorithm later.
Our process may add some extra nodes to $\Vst$ (in the unlikely event that a node supposed to get permanent slack does not obtain it) and remove some nodes from it (in the event that a node of $\Vst$ gets some permanent slack) compared to the original, fixed definition of $\Vst$ as a set of nodes adjacent to many nodes likely to get permanent slack, but it guarantees nonetheless what matters, that every node gets slack (temporary or permanent) with probability $1-d^{-\omega(1)}$.

\begin{restatable}{proposition}{slacklemmasparse}[Proposition 1 in \cite{HKNT21}]\label{L:slackgen-sparse}
Assume all nodes have degree at least $s \geq C\cdot \ln^2 \Delta$ for some universal constant $C$. There is an $O(1)$-round procedure that identifies a subset $\Vst\subseteq \Vsp$ such that after running {\slackgeneration} in the subgraph induced by $\Vsp \cup \Vun$:
\begin{compactenum}
    \item Each node $v$ in $\Vst$ has $\Omega(d_v)$ uncolored neighbors in $V \setminus \Vst$ w.p.\ $1 - \exp(-\Omega(d_v))$, and
    \item Each node $v$ in $\Vun \cup \Vsp\setminus \Vst$ has slack $\Omega(d_v)$, 
    w.p.\ $1 - \exp(-\Omega(\sqrt{s}))$.
\end{compactenum}
For each node, the probability bounds hold even when conditioned on arbitrary random choices outside its 2-hop neighborhood.
\end{restatable}

Sparse and uneven nodes then run a procedure \slackcolor that, given nodes each with slack $s_v \in \Omega(d_v)$ and $s_v \geq \smin$, where $\smin$ is a lower bound of the slack of every participating nodes known by all of them, colors them in $O(\log^* \smin)$ rounds with probability 
$1-\exp(-\smin^{\Omega(1)})-\Delta\exp(-\Omega(\smin))$. The procedure directly works in \CONGEST if we can give a \CONGEST implement of its main subroutine, \multitrial, which was done in \cref{sec:multitrial}. In \LOCAL, this subroutine simply consists of each node trying $x$ random colors from its palette, which, when the nodes have the needed slack, results in each node getting colored w.p.\ $1-\exp(-\Omega(x))$ (as if each color had an independent constant success probability). The \CONGEST subroutine we give achieves essentially the same performance using representative hash functions, which allow us to (imperfectly) communicate $x$ colors in less than the na\"ive $x \log \card{\colspace}$ bits.

Dense nodes follow a more involved algorithm. In the original \LOCAL algorithm, each almost-clique elects a leader according to a metric called slackability, which combines two measures: sparsity and discrepancy (a measure of how much one's palette differs from those of one's neighbors).
To adapt this step to \CONGEST, we instead elect the leader through a different process which leverages the relative uniformity of sparsity inside each almost-clique and the fact that discrepancy can be estimated by its contribution to a node's slack. We give details of the leader selection process in \cref{sec:leader}.
The almost-clique is then partitioned into inliers (direct neighbors of the leader, sharing many of its neighbors, and of not too high degree) and outliers (other nodes). This partitioning is easily done in \CONGEST, as it only requires nodes to announce whether they are directly connected to the leader, count how many of their neighbors are direct neighbors of the leader, send this count and their degree to the leader, and let the leader pick its inliers according to the $O(\log \Delta)$-bit information it received from each of its in-clique neighbors.

Almost-cliques in which the leader has small slackability (below a threshold related to the degree range) compute so-called ``put-aside" sets, which provide temporary slack for the remaining nodes of the almost-clique. These put-aside sets are sampled by a simple procedure: each inlier joins the put-aside set of its almost-clique according to a biased coin flip, and leaves it if one of its neighbors in another almost-clique also had a positive coin flip. This is readily implemented in \CONGEST.

In addition to using \slackcolor as previously with sparse and uneven nodes, dense nodes use a procedure \synchronizedcolortrial in which the leader randomly distributes colors from its palette to uncolored inliers (not in the put-aside set). The only obstacle in \CONGEST is possibly the size of colors, which is treated in \cref{sec:large-colors}. In the final step of the randomized part of the algorithm, the leader of each almost-clique with a put-aside set learns enough of the palettes of the put-aside elements to color them. Learning enough colors from each palette is done using other nodes of the almost-clique as relays, which we explain in \cref{sec:coloring-put-aside}.

Finally, a part of our randomized algorithm that deals with nodes of degree $o(\log n)$ will likely fail at coloring some nodes. Such nodes are colored with a deterministic algorithm, following the standard shattering framework \cite{BEPSv3}. To deal with large colors in this last phase (the \emph{post-shattering} phase), we compute a network decomposition and have each component compute a hash function without any collision in each node's palette to reduce the space of colors. This allows us to apply a deterministic algorithm whose runtime depends on the size of the space of colors. We explain this in \cref{sec:large-colors}.

A full pseudocode description of the algorithm is available in \cref{app:deg+1_alg} for completeness. 

\subsection{Leader Selection}
\label{sec:leader}

The original \LOCAL algorithm takes as leader of an almost-clique $C$ the node $w:=\arg\min_{v\in C} \barsigma_v$ of minimum slackability within $C$. In addition, it defines the slackability $\barsigma_C$ of an almost-clique $C$ as the slackability of this minimal node. Both selecting the leader properly and estimating its slackability accurately are important for the algorithm, as the leader has unique duties within the clique that not all nodes of the clique are fitted for, and cliques are assigned different behaviors depending on their slackability.

Computing the slackability of a node exactly would be too expensive in \CONGEST, as would computing the slackabilities of all the nodes to find the node with the minimum value. The slackability might also be hard to estimate accurately in low-slackability almost-cliques (in very much the same way that our procedure \estimatesparsity does not give accurate results for nodes of sublinear sparsity).

Fortunately, taking as leader the node of minimum slackability is not necessary. What is necessary is selecting a leader whose slackability is of the same order of magnitude or less than the amount of slack nodes in the almost-clique later have when running \slackcolor. This suffices as the number of nodes staying uncolored after \synchronizedcolortrial is bounded by the slackability of the leader, in expectation. As nodes get slack from a different source depending on whether they are in a low- or high-slack almost-clique, what constitutes a good leader differs slightly between the two types of almost-cliques. It is also important that almost-cliques estimate their slackability sufficiently well so that a very low-slack almost-clique does not consider itself high-slack, or vice-versa.

In low-slack almost-cliques ($\barsigma_C \leq \ell = \log^{2.1} \Delta$), it suffices that we select a leader of slackability $O(\ell)$. This leader might be significantly worse than the true leader $w$ at coloring the almost-clique when running \synchronizedcolortrial. However, it will bring down the number of uncolored nodes to $O(\ell)$, which is sufficient, as put-aside sets provide $\Omega(\ell)$ slack in low-slack almost-cliques. Second, in almost-cliques of higher slack, it suffices that the leader we pick has slackability of order $O(\barsigma_C)$ instead of exactly $\barsigma_C$. As in the low-slack case, the leader will be worst at coloring the almost-clique during \synchronizedcolortrial, but a leader of slackability $O(\barsigma_C)$ preserves that \synchronizedcolortrial likely colors all but $O(\barsigma_C)$ nodes of the almost-clique, which guarantees that the nodes have slack linear in their uncolored degree.

We show that good-enough leader can be selected through a combination of three metrics: anti-degree, external-degree, and a quantity we call \emph{chromatic slack}.

\begin{definition}[Chromatic slack]
\label{def:chromatic-slack}
Let $v$ be a dense node in an almost-clique $C$. Its (in-clique) \emph{chromatic slack} $\kappa_v$ is defined as the number of $v$'s neighbors that adopted a permanent color outside of $v$'s original palette $\pal_v$ during \slackgeneration.
\end{definition}

Our method for selecting a good-enough leader is summarized in \cref{lem:leader-selection}.

\begin{lemma}
\label{lem:leader-selection}
For an almost-clique $C$, let it pick as leader the node $x$: $x = \arg\min_{v \in C} (e_v + a_v + \kappa_v)$. Then:
\begin{itemize}
    \item If $C$ is high-slack, $x$ has slackability $\barsigma_x \in O(\barsigma_C)$, w.p.\ $1-\exp(-\Omega(\barsigma_C))$.
    \item If $C$ is low-slack, $x$ has slackability $\barsigma_x \in O(\ell)$, w.p.\ $1-\exp(-\Omega(\ell))$.
\end{itemize}
\end{lemma}

We prove \cref{lem:leader-selection} through the combination of two structural results from previous works (\cref{lem:low-ext-degree,lem:low-anti-degree}) and a statement on the distribution of chromatic slack.

\begin{lemma}[Lemma 2 in \cite{HKNT21}]
\label{lem:low-ext-degree}
There is a constant $\cext=\cext(\eacd)$
such that $e_v\le \cext\cdot \sigma_v$ holds for every node $v$ in an almost-clique $C$. 
\end{lemma}
\begin{lemma}[Lemma 3 in \cite{HKNT21}]
\label{lem:low-anti-degree}
There is a constant $\cant=\cant(\eacd)$
such that $a_v\le \cant\cdot \sigma_v$ holds for any dense node $v$.
\end{lemma}

Let the in-clique discrepancy of a node $v$ be defined as $\discC_v := \sum_{u \in N_C(v)} \disc_{u,v} = \sum_{u \in N_C(v)} \card{\pal_u \setminus \pal_v} / \card{\pal_u}$.

\begin{lemma}
\label{lem:chromatic-slack}
Consider a node $v \in C$ and $\mu_L,\mu_H$ such that $\mu_L \leq \discC_v \leq \mu_H$. During \slackgeneration, $v$ gets chromatic slack $\kappa_v \in O(\mu_H)$ w.p.\ $1-\exp(-\mu_H)$, and $\kappa_v \in \Omega(\mu_L)$ w.p.\ $1-\exp(-\mu_L)$.
\end{lemma}
\begin{proof}
Let $\pgen$ be the (constant) probability that a node tries a color during \slackgeneration.
Let us define the random variables $X_u$, $Y_u$, $Z_\col$ and their sums $X$, $Y$, $Z$ as follows:
\begin{compactitem}
    \item For each $u \in N_C(v)$, $X_u$ corresponds to the event that $u$ tries a color outside $v$'s palette.
    \item For each $u \in N_C(v)$, $Y_u$ is the same event as $X_u$, with the addition that $u$ gets to keeps the color it tries as permanent color.
    \item For each $\col \in \colspace \setminus \pal_v$, $Z_\col$ is the event that a unique node in $N_C(v)$ tries $\col$ and keeps it as permanent color.
    \item  $X = \sum_{u \in N_C(v)} X_u$, $Y = \sum_{u \in N_C(v)} Y_u$, and $Z = \sum_{u \in \colspace \setminus \pal_v} Z_u$. 
\end{compactitem}

The upper bound follows directly from $\kappa_v = Y$ and $\Exp[Y] \leq \Exp[X] = \pgen\discC_v \leq \pgen \mu_H$. Since the events $X_u$ are all independent, by Chernoff, $\kappa_v \in O(\mu_H)$ w.p.\ $1-\exp(-\Omega(\mu_H))$.

For the lower bound, we relate the chromatic slack to $Z$. First $\kappa_v = Y \geq Z$ simply by definition, and $\Exp[Z] \in \Omega(\Exp[X])$ as each color try has an $\Omega(1)$ chance of being successful and unique within $C$. Since $\Exp[X] = \pgen \discC_v \geq \pgen \mu_L$, $\Exp[Z] \in \Omega(\mu_L)$.
By \cref{lem:talagrand-difference}, as $Z$ is the difference of two $O(1)$-Lipschitz and $O(1)$-certifiable random quantities, $Z \in \Omega(\Exp[Z])$ w.p.\ $1-\exp(-\Omega(\Exp[Z]))$. Hence, $\kappa_v \in \Omega(\mu_L)$ w.p.\ $1-\exp(-\Omega(\mu_L))$.
\end{proof}

\begin{proof}[Proof of \cref{lem:leader-selection}]
By definition, $e_x + a_x + \kappa_x \leq e_w + a_w + \kappa_w$, where $w$ is the node of minimum slackability within $C$. \Cref{lem:low-ext-degree,lem:low-anti-degree} imply $(e_w + a_w)\in O(\barsigma_C)$. By \cref{lem:chromatic-slack}, since $\discC_w \leq \barsigma_w = \barsigma_C$ by definition, $\kappa_w \in O(\ell)$ w.p.\ $1-\exp(-\Omega(\ell))$ when $C$ is low-slack ($\barsigma_C \leq \ell$), and $\kappa_w \in O(\barsigma_c)$ w.p.\ $1-\exp(-\Omega(\barsigma_C))$ when $C$ is high-slack. This implies overall that for the leader $x$ selected:
\begin{compactitem}
    \item $(e_x + a_x + \kappa_x) \in O(\ell)$ w.p.\ $1-\exp(-\Omega(\ell))$ if $C$ is low-slack,
    \item $(e_x + a_x + \kappa_x) \in O(\barsigma_C)$ w.p.\ $1-\exp(-\Omega(\barsigma_C))$ if $C$ is high-slack.
\end{compactitem}

From now on, let us focus on the high-slack case, the low-slack case being similar. We have $(e_x + a_x + e_w + a_w) \in O(\barsigma_C)$, which means that $x$'s and $w$'s neighborhood may not significantly differ, i.e., $|N(x) \mathbin{\triangle} N(w)| \in O(\barsigma_C)$. Their sparsities therefore only differ by $O(\barsigma_C)$. Their discrepancies also only differ by $O(\barsigma_C)$, as $\kappa_x \in O(\barsigma_C)$ means that $\discC_x \in O(\barsigma_C)$, as a larger in-clique discrepancy would have likely resulted in much higher chromatic slack by \cref{lem:chromatic-slack}, and other differences in discrepancy must come from difference in neighborhoods, which we have shown to be bounded by $O(\barsigma_C)$.
\end{proof}

It only remains for the almost-clique to estimate its slackability with enough accuracy to categorize itself as either high- or low-slack. The aggregate $(e_x + a_x + \kappa_x)$ we used to pick a leader gives us some idea of the slackability of the almost-clique, in that it is upper-bounded by $O(\max(\ell,\barsigma_C))$ as we have seen in the proof of \cref{lem:leader-selection}. However, it does not measure in-clique sparsity. For instance, the aggregate could even be $0$ in a high-slack almost-clique $C$: it suffices that $C$'s slackability is mostly due to sparsity, and that it contains a node $v\in C$ that is connected to all other nodes of $C$, has no external neighbor, and has a palette containing the palettes of other nodes in $C$.

To estimate the sparsity of the almost-clique, we approximate the sparsity of the leader $x$ by counting the number of edges in its in-clique neighborhood. This is easily done in \CONGEST by having each node tell their neighbors whether they are adjacent to $x$, and having each neighbor of $x$ count and transmit to $x$ to how many neighbors of $x$ it is connected.

\begin{lemma}
Let $\hat m =m(N_C(x))=\frac{1}{2}\sum_{u\in N_C(x)}|N(u)\cap N_C(x)|$ count the number of edges in $x$'s in-clique neighborhood. Then $\hat \spar_x = \frac{1}{d_x} (\binom{d_x}{2}) - \hat m$ satisfies $\hat \spar_x \in [\spar_x,\spar_x+e_x]$, and so $(e_x + \hat \spar_x + \kappa_x) \in \Omega(\barsigma_C)$, w.p.\ $1-\exp(-\Omega(\barsigma_C))$.
\end{lemma}
\begin{proof}
The sparsity of $x$ corresponds to the number of missing edges in its neighborhood, divided by $d_x$. As $\hat m$ only counts edges in $x$'s in-clique neighborhood, it undercounts the number of edges in $N(x)$ by $e_x\cdot d_x$ or less, i.e., $\hat m \in [m(N(x)) - e_x\cdot d_x,m(N(x))]$. In turn, this means that the estimate $\hat \spar_x$ of $x$'s sparsity is in the range $[\spar_x,\spar_x+e_x]$.

The discrepancy of $x$ is the sum of the contribution of its in-clique neighbors and that of its external neighbors, that is, $\disc_x \in [\discC_x,\discC_x + e_x]$.

The slackability of $x$ is defined as $\barsigma_x = \spar_x + \disc_x$. If $\spar_x \geq \barsigma_x/3$ or $e_x \geq \barsigma_x/3$, then $(e_x + \hat \spar_x + \kappa_x) \in \Omega(\barsigma_x)$ trivially. Otherwise, $\discC_x \geq \disc_x - e_x \geq \barsigma_x/3$, which implies by \cref{lem:chromatic-slack} that $\kappa_x \in \Omega(\barsigma_x)$ w.p.\ $1-\exp(-\Omega(\barsigma_x))$. The statement can be reformulated with $\barsigma_C$ instead of $\barsigma_x$ as $\barsigma_x \geq \barsigma_C$ by definition of $\barsigma_C$.
\end{proof}

Putting everything together, we get that $(e_x + \hat \spar_x + \kappa_x) \in \Theta(\barsigma_C)$ w.p.\ $1-\exp(-\Omega(\barsigma_C))$ in high slack almost-cliques, and that $(e_x + \hat \spar_x + \kappa_x) \in O(\ell)$ w.p.\ $1-\exp(-\Omega(\ell))$ in low-slack almost-cliques, which is sufficient for our purposes.

\subsection{Coloring the Put-Aside Sets}
\label{sec:coloring-put-aside}

The coloring of the put-asides sets $P_C$ is the only step of the algorithm for dense nodes (\cref{alg:logstar-dense}) that remains to be explained.
In this step of the algorithm, the nodes in each put-aside set $P_C$ transmit the content of their palettes and the topology of $G[P_C]$ to their leader $x_C$. Provided with this information, the leader can then assign each node of $P_C$ a color from its individual palette that does not conflict with the color of its neighbors. Without any adjustment, this process has each node from a put-aside $P_C$ send $\Theta(|P_C|(\log|\colspace| + \log n))$ bits to its leader, which on a single communication link would be too costly in \CONGEST. We adapt this part of the algorithm to \CONGEST through three avenues. First, we have nodes send colors to their leader by using a hash function chosen by said leader. This solves the bandwidth requirements that sending very large colors presents, as explained in \cref{sec:large-colors}. Second, we reduce the size of the put-aside sets to just the size that is needed to get sufficient slack. Finally, we use other nodes of the almost-clique as relays to increase the bandwidth between the leader and each node of the put-aside.

The leader restricts the size of $P_C$ to $\Theta(\ell)$, which is a sufficient amount of slack for the parts of the algorithm that rely on slack from the put-aside sets (invoking \slackcolor).
Recall that $\core_C \subseteq N_C(x_C)$.
The leader enumerates the nodes in $\core_C$
and allocates each node $v\in P_C$ a contiguous interval of $2|P_C|+1$ indices, corresponding to a set $R_v$ of nodes.
Since $|\core_C| \ge 2|P_C|^2+|P_C|$, the nodes receive disjoint intervals.
Each node $v\in P_C$ has $a_v = O(\sigma_C)=O(\ell)\le |P_C|$ non-neighbors in $C$, and hence it has at least $|R_v| - a_v \ge |P_C|$ neighbors in $R_v$.
Now $v$ can send $|N(v)\cap P_C|+1$ colors from its palette to $x_C$ in $O(1)$ rounds, via the relay nodes in $N(v)\cap R_v$. The topology of $P_C$ can similarly be transmitted. The leader can then properly color $P_C$ locally and forward the colors to the nodes.

\subsection{Large Colors}
\label{sec:large-colors}

We have implicitly assumed until now that sending a color over an edge, as nodes do when broadcasting their permanent color to their neighbors, only takes $O(1)$ rounds. This is possible if the color space $\colSpace$ is of size $\card{\colSpace} \in n^{O(1)}$. In \cref{lem:representative_hash_functions}, the dependency of $t$ in $\card{\colSpace}$ is only $\log \card{\colSpace}$, meaning that sending a representative hash function still takes only $O(1)$ rounds even when $\card{\colSpace} \in \exp(n^{\Theta(1)})$. Can we tolerate such a large color space in other parts of the algorithm? We resolve this in the affirmative.

\paragraph{Pre-shattering phase.}
For all parts of the algorithm except the post-shattering phase, we achieve this using a family $\HFset$ of $1+\eps$-approximately universal hash functions, i.e., a set of hash functions $h:[N] \rightarrow [M]$ such that for all $x_1 \neq x_2$, $\Pr_{h\gets\HFset}\event{h(x_1)=h(x_2)}\leq (1+\eps) / M$. There exists small enough families of such hash functions so that specifying an element in the family only takes $O(\log\log N + \log M + \log(1/\eps))$ bits (\cite{BJKS93}, or Problem 3.4 in~\cite{Vadhan12}). Set $\eps=1$ and let us hash to $M=\Theta(n^{d})$ values, where $d \in \Theta(1)$. Under these assumptions, sending a hash value only takes $O(1)$ rounds, and sending an element of $\HFset$ takes $O(\ceil{\log\log\colSpace / \log n})$ rounds -- in particular, $O(1)$ if colors are written on $\poly(n)$ bits. Let each node $v$ pick and broadcast a random $1+\eps$-approximately universal hash function $h_v$ from $\HFset$ at the start of our algorithms. Whenever a node $u$ was previously sending a color $\col$ to a node $v$ in our algorithms, we now have $u$ send $h_v(\col)$ to $v$. Granted no collision occurs in any neighborhood, these hash values perfectly replace the actual colors wherever nodes were previously using the exact colors of their neighbors, such as when updating their palettes, computing their chromatic slack, and when a leader in an almost-clique sends colors to the inliers -- each inlier looks for a color that hashes to the hash sent by the leader, and then tries that color by hashing it using its neighbors' hash functions.

With $\log n$ bandwidth, we ensure no collision occurs in any neighborhood w.h.p., by taking $d$ appropriately large. Consider a node $v$ and its neighborhood. There are at most $(\Delta+1)^2$ distinct colors in the palettes of $v \cup N(v)$. The probability that a collision occurs in these colors with a random hash function from $\HFset$ is bounded by $\binom{(\Delta+1)^2}{2} \cdot n^{-d} \leq n^{-d+4}$. So, w.p.\ at least $1-n^{-d+5}$, there are no collisions in all neighborhoods. Setting $d \ge 6$, this holds w.h.p.

\paragraph{Post-shattering phase.}
For the post-shattering phase, unlike in the \LOCAL model, in general we may not directly use one of the recent deterministic algorithms of Ghaffari and Kuhn~\cite{GK21}, as the complexity of their algorithm in \CONGEST depends on the size of the color space $\colspace$. Indeed, their two \CONGEST algorithms use either $\log^2 \Delta \log n$ rounds of $\Delta \log \card{\colspace}$ bits or $\log^2 \card{\colspace} \log n$ rounds of $\log \card{\colspace}$ bits (note that $n$ and $\Delta$ are the parameters of the shattered connected component, but $\colspace$ is the original color space). When $\card{\colspace} \in \poly(\log n)$ (and therefore, all the degrees of the graph as well) we get an $O(\log^3 \log n)$ algorithm, but handling a larger color space requires additional work. 

We handle larger color spaces by computing a network decomposition on the shattered graph in $O(\log^5 \log n)$ rounds \cite{GGR20}, and coloring each cluster of each color class by first computing a color space reduction before using the deterministic algorithm of \cite{GK21}. The color space reduction simply consists of finding a function that maps each color from $\colspace$ to a $\poly(\log n)$ number such that no collision occurs in any node's palette. This is achieved through derandomizing the random selection of such a function with the method of conditional expectation.

\begin{lemma}[Lemma 3.19 in \cite{HKMN20} (full version)]
Let $N \in \log^{O(1)} n$ be the size of the subgraph on which we compute a network decomposition of diameter $D$.
Consider one cluster $C$ of the network decomposition and let $\pal(u)$ be the palette of vertex $u \in C$ of size $L \leq N$. There is a deterministic $D\cdot \log N$ round algorithm that computes a colorspace reduction $f : \colspace \rightarrow N^{10}$ such that $\card{f(\pal(u))} = \card{\pal(u)}$ for all $u \in C$.
The colorspace reduction $f$ can be described with $O(\log \log n)$ bits.
\end{lemma}

\section{Full Statement of (Degree+1)-List-Coloring Algorithm}
\label{app:deg+1_alg}

\subsection{Broad Structure}

As explained in the algorithm overview in \cref{app:details}, nodes are dealt with in degree ranges of the form $[\log^7 x,x]$. The full algorithm simply consists of $O(\log^* n)$ call to \cref{alg:logstar-both}, which assumes the degrees of nodes in the graph to be within such a degree range.

\begin{algorithm}[H]
\caption{Randomized $\degoLC$ Algorithm ($\forall v, d_v \in[\log^7 \Delta,\Delta]$)}
\label{alg:logstar-both}
  \begin{algorithmic}[1]
    \STATE \computeacd.
    \STATE Apply \cref{alg:logstar-sparse} to sparse nodes.
    \STATE Apply \cref{alg:logstar-dense} to dense nodes.
  \end{algorithmic}
\end{algorithm}

The algorithm for a given degree range (\cref{alg:logstar-both}) itself calls two procedures: one that colors the sparse nodes (\cref{alg:logstar-sparse}), and one that colors the dense nodes (\cref{alg:logstar-dense}). Before that, it computes an almost-clique decomposition, which we explained how to do in \CONGEST in \cref{sec:acd}.

\begin{algorithm}[H]
\caption{Main Procedure for Coloring Sparse Nodes}
\label{alg:logstar-sparse}
    \begin{algorithmic}[1]
    \STATE Identify the set $\Vst \subset \Vsp$
    \STATE {\slackgeneration} in $G[\Vsp\cup \Vun]$.
    \STATE {\slackcolor} $\Vst$. \label{st:sp-multitrial}
    \STATE {\slackcolor} $\Vsp \setminus \Vst$ and $\Vun$. \label{st:o-multitrial2}
\end{algorithmic}
\end{algorithm}

\begin{algorithm}[H]
\caption{Main Procedure for Coloring Dense Nodes}
 \label{alg:logstar-dense}
  \begin{algorithmic}[1]
  \STATE Compute the leader $x_C$ and outliers $O_C$ of each almost-clique $C$. Let $O = \cup_C O_C$. \label{st:outliers-hi}
  \STATE {\slackgeneration}.
  \STATE $P_C \gets \putaside[(C)]$ in each low-slack almost-clique $C$. Let $P = \cup_C P_C$.      \label{st:putaside-hi}
  \STATE {\slackcolor} $O$. \label{st:o-multitrial}
    \STATE {\synchronizedcolortrial} $\Vdense \setminus P$.\label{st:synchtrial}
    \STATE {\slackcolor} $\Vdense \setminus P$. \label{st:lastmultitrial-hi}
    \STATE For each low-slack $C$, let $x_C$ collect the palettes in $P_C$ and color the nodes locally.\label{st:collect-hi}
\end{algorithmic}
\end{algorithm}

Approaching the $\degoLC$ problem by giving two separate algorithms for sparse and dense nodes is natural. Indeed, splitting the problem in that manner results in two \degoLC instances of similar degree ranges (more generally, the problem is self-reducible), and the all-sparse and all-dense cases are possible inputs that need to be considered anyway.

\subsection{Subroutines}

We now detail the subroutines referred to in \cref{alg:logstar-sparse,alg:logstar-dense} above.

\paragraph{Trying colors and slack generation. } \slackgeneration simply consists of each node trying a random color in its palette, with some constant probability. What is trying a color formally means is described in \cref{alg:trycolor}, in particular what it means to try a color when some nodes have priority over other nodes. How trying a color can be done in \CONGEST even when the color space is of order $\exp(n^{\Theta(1)})$ is explained in \cref{sec:large-colors}.

\begin{algorithm}[H]\caption{\slackgeneration[(probability $\pgen$)]}\label{alg:slackgeneration}
\begin{algorithmic}[1]
\STATE $S\gets $ sample each $v\in G$ into $S$ independently w.p.\ $\pgen=1/10$.
\STATE \algorithmicforall\ $v\in S$ in parallel \algorithmicdo\ {\tryrandomcolor}$(v)$.
\end{algorithmic}
\end{algorithm}

\begin{algorithm}[H]\caption{\tryrandomcolor (vertex $v$)}\label{alg:tryrandomcolor}
\begin{algorithmic}[1]
\STATE Pick $\col_v$ u.a.r.\ from $\pal_v$. 
\STATE $\trycolor(v,\col_v)$
\end{algorithmic}
\end{algorithm}

A more refined version gives priority to some nodes over others: for each node $v$, we partition its neighborhood $N(v)$ into $N^+(v)$ -- the nodes whose colors conflict with $v$'s -- and $N^-(v) = N(v)\setminus N^+(v)$. For correctness of \trycolor, $u \in N^-(v) \rightarrow v\in N^+(u)$ should hold for each edge $uv$. 

\begin{algorithm}[H]\caption{\trycolor (vertex $v$, color $\col_v$)}\label{alg:trycolor}
\begin{algorithmic}[1]
\STATE Send $\col_v$ to $N(v)$, receive the set $T^+=\{\col_u : u\in N^+(v)\}$.
\STATE{\textbf{if}} $\col_v\notin T^+$ \textbf{then} permanently color $v$ with $\col_v$.
\STATE Send/receive permanent colors, and remove the received ones from $\pal(v)$.
\end{algorithmic}
\end{algorithm}

\paragraph{Leader, inliers, and outliers of an almost-clique.}
Once the leader $x$ of an almost-clique $C$ is chosen, the \emph{outliers} $O_C$ of this almost-clique are chosen to be:
\begin{compactenum}
    \item the $\max(d_x,\card{C})/3$ nodes in $C$ with the fewest common neighbors with $x$,
    \item the $\card{C}/6$ nodes of largest (original) degree, and 
    \item the anti-neighbors $A_x$ of $x$.    
\end{compactenum}

The \emph{inliers} $I_C$ are the rest of the almost-clique, $I_C=C\setminus O_C$.
Recall that $\ell = \log^{2.1} \Delta$. How the leader selection process is adapted to work in \CONGEST is explained in \cref{sec:leader}.

\paragraph{Put-aside sets.} The construction of put-aside sets is a simple random sample (\cref{alg:pas}). How these sets are colored at the end of the procedure for dense nodes in \CONGEST is explained in \cref{sec:coloring-put-aside}.

\begin{algorithm}[H]\caption{{\putaside}$(C)$} 
\label{alg:pas}
  \begin{algorithmic}[1]
  \STATE $S_C\gets$ each node $v\in \core_C$ is sampled independently w.p.\ $\pdisj = \ell^2/(48\Delta_C)$.
  \RETURN $P_C \gets \{v\in S_C : E_v\cap S = \emptyset\}$, where $S = \cup_{C'} S_{C'}$
 \end{algorithmic}
\end{algorithm}

\paragraph{Synchronized color trials within almost-cliques. }

An important subroutine of the algorithm for dense nodes is \synchronizedcolortrial, in which the leader of each almost-clique $C$ randomly gives a unique color from its palette to the uncolored non-put-aside inliers of $C$. The only possible issue in \CONGEST is that the colors may be too large to send efficiently. How to overcome this hurdle is explained in \cref{sec:large-colors}.

\begin{algorithm}[H]\caption{\synchronizedcolortrial, for almost-clique $C$}
\label{alg:synchtrial}
  \begin{algorithmic}[1]
    \STATE $x_C$ randomly permutes its palette $\pal(x_C)$, sends each neighbor $u \in \core_C$ a distinct color $\col_u$. \label{st:randomorder}
    \STATE Each $u \in \core_C$ calls {\trycolor}($u$, $\col_u$) if $\col_u \in \pal(u)$
    \end{algorithmic}
\end{algorithm}

\paragraph{Coloring with slack. }
Finally, we give the pseudocode for \slackcolor, an important subroutine in all randomized algorithms that achieve complexity $O(\log^* n)$ for graphs of large enough degree without increasing the number of colors polynomially as in Linial's algorithm. This subroutine has nodes with slack linear in their degree try increasing numbers of colors through $O(\log^*n)$ iterations. How to implement in \CONGEST its main building block, \multitrial, is explained in \cref{sec:multitrial}. The nodes can also readily compute their slack in \CONGEST with the techniques for handling large colors described in \cref{sec:large-colors}, making the whole procedure implementable in \CONGEST.

$\kappa\in (1/\smin,1]$ is a parameter, $a \knuthupuparrow b$ denotes tetration ($a \knuthupuparrow 0 = 1$, $a \knuthupuparrow (b+1) = a^{a \knuthupuparrow b}$).

\begin{algorithm}[H]\caption{\slackcolor[($\smin$)], for node $v$} 
\label{alg:slackcoloring}
  \begin{algorithmic}[1]
  \STATE \algorithmicfor\ $O(1)$ rounds \algorithmicdo\  {\tryrandomcolor}($v$).\label{step:slackcolor-begin-init} 
    \STATE \algorithmicif\ $s(v) < 2d(v)$ \algorithmicthen\ terminate.\label{step:slackcolor-end-init}
    \STATE Let $\sminpow\gets \smin^{1/(1+\kappa)}$
    \FOR{$i$ from $0$ to $ \log^* \sminpow$}\label{step:slackcolor-begin-tower}
    \STATE $x_i \gets 2 \knuthupuparrow i$ 
    \STATE $\multitrial(x_i)$ 2 times.
    \STATE \algorithmicif\ $d(v) > s(v) / \min(2^{x_i},\sminpow^{\kappa})$ \algorithmicthen\ terminate.\label{step:slackcolor-termtower}
    \ENDFOR\label{step:slackcolor-end-tower}
    \FOR{$i$ from $1$ to $\ceil*{1/\kappa}$}\label{step:slackcolor-begin-finish}
    \STATE $x_i \gets \sminpow^{i \cdot \kappa}$ 
    \STATE $\multitrial(x_i)$ 3 times.
    \STATE \algorithmicif\ $d(v) > s(v) / \min(\sminpow^{(i+1)\cdot\kappa},\sminpow)$ \algorithmicthen\ terminate.\label{step:slackcolor-termfinishloop}
    \ENDFOR
    \STATE $\multitrial(\sminpow)$.\label{step:slackcolor-end-finish}
\end{algorithmic}
\end{algorithm}

\end{document}